\documentclass[sigconf]{acmart}
\pdfoutput=1

\setcopyright{acmlicensed}
\renewcommand\footnotetextcopyrightpermission[1]{}

\copyrightyear{2020}
\acmYear{2020}

\acmConference[DEBS '20]{The 14th ACM International Conference on Distributed and Event-based Systems}{July 13--17, 2020}{Virtual Event, QC, Canada}
\acmBooktitle{The 14th ACM International Conference on Distributed and Event-based Systems (DEBS '20), July 13--17, 2020, Virtual Event, QC, Canada}

\usepackage[T1]{fontenc}
\usepackage[utf8]{inputenc}
\usepackage{verbatim}
\usepackage{array}
\usepackage{dblfloatfix}
\usepackage{algorithmic}
\usepackage{graphicx}
\graphicspath{{graphics/}{moregraphics/}}
\DeclareGraphicsExtensions{.pdf,.PDF,.jpeg,.JPEG,.jpg,.JPEG,.png,.PNG}

\usepackage{booktabs}

\usepackage{siunitx}

\usepackage{textcomp}
\usepackage{xcolor}
\usepackage{mymacros}

\usepackage{xspace}
\newcommand{\platform}{\textsc{MODiCuM}\xspace}

\newcommand{\field}[1]{\textnormal{\textit{#1}}}

\begin{document}



\title[Mechanisms for Outsourcing Computation via a Decentralized Market]{Mechanisms for Outsourcing Computation via a\\Decentralized Market}

\author{Scott Eisele}
\affiliation{Vanderbilt University}
\author{Taha Eghtesad}
\affiliation{University of Houston}
\author{Nicholas Troutman}
\affiliation{University of Houston}
\author{Aron Laszka}
\affiliation{University of Houston}
\author{Abhishek Dubey}
\affiliation{Vanderbilt University}





\begin{abstract}
As the number of personal computing and IoT devices grows rapidly, so does the amount of computational power that is available at the edge.
Since many of these devices are often idle, there is a vast amount of computational power that is currently untapped, and which could be used for outsourcing computation. 
Existing solutions for harnessing this power, such as volunteer computing (e.g., BOINC), are centralized platforms in which a single organization or company can control participation and pricing. By contrast, an open market of computational resources, where resource owners and resource users trade directly with each other, could lead to greater participation and more competitive pricing. To provide an open market, we introduce \platform, a decentralized system for outsourcing computation. \platform deters participants from misbehaving---which is a key problem in decentralized systems---by resolving disputes via dedicated mediators and by imposing enforceable fines. However, unlike other decentralized outsourcing solutions, \platform minimizes computational overhead since it does not require global trust in mediation results. We provide analytical results proving that \platform can deter misbehavior, and we evaluate the overhead of \platform using experimental results based on an implementation of our~platform.

\end{abstract}





\ifFull
\noindent{Published in the proceedings of the 14th ACM International Conference on Distributed and Event Based Systems (DEBS 2020).}
\fi
\ifFull{{\let\newpage\relax\maketitle}}
\else{\maketitle}
\fi

\section{Introduction}
\label{sec:intro}


The number of computing devices---and thus computational power---available at the edge is growing rapidly; this trend is projected to continue in the future~\cite{Analyst2016IoTP}. Many of these are end-user or IoT devices that are often idle since they  were installed for a specific purpose, which they can serve without using their full computational power. 
Our goal is to  harness these untapped computational resources by creating an open market for outsourcing computation to idle devices. Such a market would benefit device owners since they would receive payments for 
computation while incurring negligible costs. \revision{To illustrate, running an AWS Lambda instance with 512MB of memory for \camera{1-hour} costs \$0.03, while the electrical cost of operating a BeagleBone Black\footnote{\revision{\url{https://github.com/beagleboard/beaglebone-black/wiki/System-Reference-Manual}}}
single-board computer with 512MB of memory for an hour is 100 times less. 
}
\camera{Thus, it is feasible that a computation service could be provided economically.}

Prior efforts to leverage these underutilized resources include \emph{volunteer computing} projects, such as BOINC~\cite{anderson_designing_2006} and CMS@Home \cite{field2018cms}, in which users donate the computational resources of their personal devices to be used for scientific computation. Volunteer computing suffers from two limitations that prevent it from broader utility. First, the resources made available by volunteer computing participants are only accessible to specific users and projects. Second, it relies on systems ``volunteering'' their time as it  does not include incentives to provide reliable access to computational resources,  leading to the problem of low participation~\cite{nouman_durrani_volunteer_2014}. 

Participation can be incentivized through the implementation of a competitive market, which facilitates the discovery and allocation of supply and demand for computational resources and tasks. 
The market must provide mechanisms to address misbehavior and resolve any disputes\footnote{Disputes are disagreement between the parties about the correctness of the job execution. They may arise due to
a fault or malicious behavior.} between the participants. Such a market
could be managed by a central organization, as many in the sharing economy are (e.g., Uber, Airbnb). A central organization could mediate disputes. However, a centralized system presents a clear target for attackers, can be a single point of failure, and without competition may charge exorbitant fees. An alternative is to create an open and decentralized  market, where resource owners and resource users trade directly with each other, which could lead to greater participation, more competitive pricing, and improved~reliability.

In distributed computing systems, faults and misbehavior have traditionally been addressed using consensus algorithms. Recently, distributed ledgers  have emerged as a novel mechanism to provide consensus in decentralized public systems~\cite{uriarteBlockchain2018}. 
Smart contracts extend the capabilities of a distributed ledger by enabling ``trustless'' computation on the stored data.
In theory, smart contract implementations, such as the one found in  Ethereum~\cite{underwood2016blockchain,wood2014ethereum}, could be used for outsourcing complex computations. However, since the computation is replicated on thousands of nodes, it becomes costly. To reduce costs, complex computations must be executed off-chain and only result aggregation, validation, and record keeping should be kept on the chain (see \cite{transax} for example).

Prior efforts to construct outsourced computation markets using distributed ledgers
include FogCoin~\cite{fogcoin}, TrueBit~\cite{teutsch2017scalable}, Golem~\cite{Golem}, and iExec~\cite{iExec}. Unfortunately, these existing solutions have varying degrees of inefficiency due to extensive verification of the correctness of the computation performed. 
There have been some efforts to ameliorate this situation.
For example, TrueBit performs computation using a typical computer and relies on the distributed ledger only to complete disputed instructions; however, this approach is still quite inefficient. We discuss  these existing solutions and their drawbacks in detail in Section~\ref{sec:related}.

{\bf Contributions:}
The key problem in implementing a decentralized market is the \emph{efficient} resolution of disputes, which includes determining if the results of outsourced jobs are correct and if resource providers should receive their payments for the jobs.
This paper introduces \emph{Mechanisms for Outsourcing via a Decentralized Computation Market}  (\platform), a distributed-ledger based platform for decentralized computation outsourcing. 
In contrast to other distributed-ledger based solutions (mentioned above), our approach retains computational efficiency 
by minimizing the amount of resources spent on verification through three ideas. 
 First, it relies on \emph{partially trusted mediators}
for settling disputes instead of trying to establish a global consensus on the results of outsourced computation, which  would require extensive duplication of computation.
Second, it \emph{verifies random subsets} of results, which keeps verification costs low while supporting a wide range of jobs.
Third, it deters misbehavior through \emph{rewards and fines}, which are enforced by a distributed-ledger based smart contract. As a result, \platform does not prevent cheating and misuse, but it deters rational agents from misbehavior. The specific contributions of this paper are as~follows:

\begin{enumerate}[noitemsep,leftmargin=*,topsep=0.2em]
\item We introduce a smart contract-based protocol and a platform architecture  for incentivizing the participation of job creators, resource providers, and mediators. 
\item We present an analysis of the protocol by modeling the exchange of resources as a game, and show how we can select the values of various parameters (fines, deposits, and rewards), ensuring that honest participants will not lose, and the advantage of dishonest participants is bounded.
\item We provide a proof-of-concept implementation of the protocol, built on top of Ethereum. Through comparison against AWS lambda we determine that due to the transaction costs currently associated with the main Ethereum blockchain, our implementation is currently suited only for very long running jobs. However, any improvements to Ethereum will benefit our platform. Additionally, our protocol is not limited to the use of a specific platform. Any platform that supports smart contract functionality can be utilized.
\end{enumerate}

{\bf Outline:} We begin by discussing related work in \cref{sec:related}. Then, we introduce \platform in \cref{sec:arch,sec:protocol}. We analyze the protocol in \cref{sec:analysis}. In \cref{sec:implementation}, we describe an implementation of our platform and provide experimental results on its performance. 
Finally, in Section~\ref{sec:concl}, we present concluding remarks.
Implementation is available at \cite{modicum}, and proofs and detailed specification can be found~\FullText{in the appendices.}{in our full paper \cite{modicum-arxiv}.}


\section{Related Work}
\label{sec:related}

In \cite{dong2017betrayal},
the authors determine that for verifying outsourced computation the cryptographic approach is not practical and should instead use repeated executions. The computations however should not be duplicated more than twice. Their strategy is simple: outsource to two providers and compare results. The key challenge then is to prevent collusion. They propose \textit{sabotaging collusion with smart contracts}.
Essentially, they hold the providers accountable using security deposits and a smart contract. They also assume that if two providers intend to collude, the providers also use a smart contract to hold each other accountable. To counter this, the authors propose a third contract that states that the first provider who betrays the other, showing the colluder's contract as proof, is granted immunity and will receive a reward.
In their work, they have not yet considered the scenario when the client may be an adversary. They also do not address the case when the contractors do not need a contract to trust each other.

The authors of \cite{belenkiy2008incentivizing} consider a case where there is a trusted third party that would be responsible for verifying a critical computation, except that it becomes a bottleneck for the rest of the system. It instead becomes a boss and outsources the verification task. To incentivize participation, the boss offers a reward, and to discourage misbehavior the boss requires a security deposit to enable it to enforce fines. The two verification strategies they consider are random double-checking by the boss and hiring multiple contractors to perform a job and comparing results. The authors do not consider the case when the boss is malicious or attempts to avoid paying out the rewards promised. They also do not discuss practical issues such as what if the contractors never return a result.

The verification protocol~\cite{Karwowski2018More} in Golem assumes that tasks are stateless and can be subdivided into sub-tasks, and it spot-checks some subset of those tasks. An example of this type of task is image rendering, which is their initial use case. They claim that they will also support other computational tasks, such as machine learning; however, their verification strategy is not guaranteed to suit such tasks. Their second, more recent protocol~\cite{2019Verification} executes each computation twice and compares the results using a globally trusted oracle that compares the two results. If there is a difference, they request an additional provider to break the tie. They do not mention collusion, which could easily break their verification protocol. Additionally, they suggest that reputation could be used to reduce the required number of repeated executions; however, 
the authors do not provide any discussion. 

In iExec \cite{Croubois2017PoCo}, the number of repeated executions required for verification is determined by a confidence threshold. After a result is computed, the pool scheduler checks if the results submitted achieve the desired level of confidence using Sarmenta's voting \cite{sarmenta2001sabotage}. 
\iflong
The iExec group provides better descriptions of their protocols and include some analysis on how much of the network attackers would need to control to cause, with some probability, an erroneous result to be accepted. They also analyze the profitability of such an attack. 
\fi
The workers register themselves to a scheduler that they choose to trust. If the scheduler breaks trust, the worker can leave to a competing pool. iExec checks tasks for non-determinism before allowing them to be deployed in the network. It does this by executing the task many times. This has the drawback that the task must take no inputs; otherwise, all execution paths would need to be tested. 

In TrueBit \cite{teutsch2017scalable}, verification is provided via Verifiers, which duplicate the computation based on an incentive structure that rewards them for finding errors, and errors are intentionally injected into the system occasionally to guarantee benefits for verification. When a Verifier finds an error, it challenges the resource to a game where they compare the machine state at various points during the execution, in a manner similar to binary search. Once the step where the two solutions deviate is found, the machine state is submitted to a mediator. Rather than using a single entity as a trusted mediator, TrueBit implements the mediation as an Ethereum smart contract which implements a virtual machine. The TrueBit virtual machine executes that specific step and determines which agent is at fault. 
\iflong
Essentially when a task is posted to the platform each resource that wants to solve it submits an Ethereum transaction, the miners choose one to include and that one becomes the assigned resource. The other resources may choose to act as verifiers but it also possible that none do. All programs must be written in Wasm~\cite{WASM} to be compatible with the on-chain mediation.
\fi
The authors recognize that even though TrueBit can theoretically process arbitrarily complex tasks, in practice the verification game is
inefficient for complex tasks. 

In each of these prior works, the authors assume that results must be globally accepted or have universal validity. In  \cite{teutsch2017scalable}, the authors mention that using a trusted mediator is an option for resolving disputes regarding whether a task was done correctly, but such a solution is unacceptable because it does not provide universal validity. We argue however that many tasks---perhaps even the majority---do not require universal validity. For such tasks, only the agent who requested the execution of the task must be convinced of the validity of the result, which means that the systems proposed by prior work incur significant and unnecessary computational overhead for such tasks. Prior works also do not address the possibility that the job creator itself  could manipulate the system by providing tasks that have non-deterministic or environment-dependent~results.



\section{\platform Architecture}
\label{sec:arch}

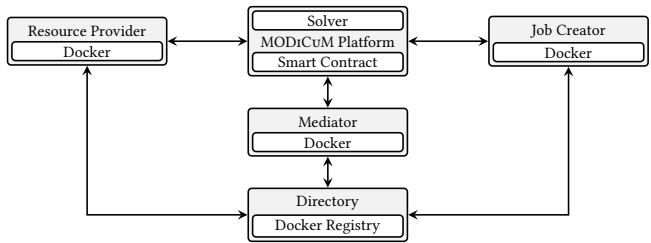
\begin{figure}[t]
    \centering
    {%
    \renewcommand{\arraystretch}{0.0}
     {%
    \begin{tikzpicture}[x=3.2cm,y=1.1cm,
        font=\scriptsize, semithick,
        Actor/.style={matrix of nodes, nodes=Component, row sep=0.05em, inner sep=0.15em, draw, fill=gray!10, minimum width=2cm, minimum height=0.5cm, align=center, rounded corners=0.15em},
        Component/.style={inner sep=0.15em, draw, minimum height=0.25cm, fill=white},
        Tag/.style={fill=none, draw=none},
        Connection/.style={<->, >=stealth}]
      \matrix [Actor] (RP) at (-1,0) {
        |[Tag]|{Resource Provider}\\
        Docker\\
      };
      \matrix [Actor] (JC) at (1,0) {
        |[Tag]|Job Creator\\
        Docker\\
      };
      \matrix [Actor] (Pla) at (0,0) {
        Solver\\
        |[Tag]|\platform Platform\\
        Smart Contract\\
      };
      \matrix [Actor] (Med) at (0,-1.1) {
        |[Tag]|Mediator\\
        Docker\\
      };
      \matrix [Actor] (Dir) at (0,-2.1) {
        |[Tag]|Directory\\
        Docker Registry\\
      };
      \draw [Connection] (RP) -- (Pla);
      \draw [Connection] (JC) -- (Pla);
      \draw [Connection] (Med) -- (Pla);
      \draw [Connection] (RP) |- (Dir);
      \draw [Connection] (JC) |- (Dir);
      \draw [Connection] (Med) -- (Dir);
    \end{tikzpicture}
    }
    }
    \caption{\platform architecture. For ease of presentation, only one instance of Resource Provider, Job Creator, Mediator, Directory, and Solver is shown.}
    \label{fig:arch}
    \vspace{-0.5cm}
\end{figure}

\platform enables the allocation and execution of computational tasks on distributed resources that may be dishonest and may enter or exit the platform at any time. 
Conceptually, this requires resource management, i.e., the allocation of resources to maximize utility. We also need a service for managing job requests, which include both the job specification and the data products related to a job. Finally, we need a service that manages the market, which includes matching jobs to available resources, tracking the provenance of job products, accounting, and handling failures. Due to the decentralization of the system, \emph{job creators} (JC) and \emph{resource providers}  (RP) can have their own strategies; however, the market must operate as a singleton. 
\iflong
Unlike a cloud service, where a single organization manages all aspects of the platform (e.g.,  accounting, log management), in our case all information must be recorded in a way that prevents tampering but allows everyone to audit system history. A distributed ledger is the state-of-the-art solution to meet these requirements. 
Another distinction from cloud services is that resource providers in \platform are independent actors, who may act selfishly and dishonestly. For example, they can advertise more resources than what is available to them, they may  fail to execute jobs correctly, and they may blame failure on the job creator. This necessitates some mechanism for verifying the execution of jobs. 
\fi



To satisfy all of these requirements, we develop an architecture that consists of a distributed-ledger based \emph{Smart Contract}, storage services (\emph{Directory}), allocation services (\emph{Solver}), \emph{Resource Providers}, \emph{Job Creators}, and \emph{Mediators}. These components and actors are shown in Figure~\ref{fig:arch} and described in the following subsections. We start by describing first what a job is in \platform.

\subsection{\platform Jobs}
\label{sec:jobs}

A job is a computing task that takes inputs and produces outputs. The jobs \platform is designed to support are limited to deterministic, batch jobs that are not time critical, and are isolated, meaning they receive no information that allows them to discern which agent (RP, M) is executing them.  We discuss the importance of our limitations in \cref{sec:mediation}. Further the confidentiality of the job cannot be guaranteed since the computation is outsourced to any capable participant.
It is essential that jobs can be executed independently of the host configuration.
The state-of-the-art approach for achieving this is to use a container technology \cite{burns2016borg}. Consequently, we use Docker \cite{merkel2014docker} images to package the jobs in our implementation. 
Docker containers provide an easy-to-use way of measuring and limiting resource consumption and securely running a process by separating the job from the underlying infrastructure. \revision{
To avoid downloading a full image for every job execution, we support Docker layers. To execute an instance of a job, a base layer (i.e., OS or framework), an execution layer with job specific code, and a data layer are required. The base layer can be downloaded a priori, and the other layers are downloaded after a match. If the execution layer has already been downloaded, then only the new data layer is required.}  The job requirements, \revision{including the required resources,} are specified in the offers that the agents~make. \revision{These requirements are used by the resource provider to set the Docker container's resource constraints at runtime.}

The cost of running a job depends on the amount of resources that the job uses. The resources we consider for feasibility are instruction count (computes from CPU speed and time on CPU), disk storage, memory, and bandwidth for downloading the job. The resources we use to compute the price are the instruction count and bandwidth used as reported by the RP where each is multiplied by the RP's asking price. Thus, the price of the job is calculated as~follows:
\begin{gather}
    \begin{aligned}
    \pi_c = & \field{result}.\field{instructionCount} \cdot 
     \field{resourceOffer}.\field{instructionPrice}  \\
    + ~ & \field{result}.\field{bandwidthUsage} \cdot
     \field{resourceOffer}.\field{bandwidthPrice}  \\
    \end{aligned}
    \label{eq:job_price}
    \raisetag{20pt}
\end{gather}


\reviewercomment{such as declaring wrong resource stats, are hard to group into wrong and right schemes and therefore pose a challenge to the implementation}


\revision{
Note that it is often difficult to estimate resource requirements precisely. Repeated executions of the same job result in approximately the same resource usage, but not exactly the same. Therefore,  when constructing an offer, the JC must add margins to the estimated resource requirements.  The upper limit of resource requirements should be set considering the maximum amount of resources that the JC is willing to pay for to account for the expected variance. If the JC did not provide this leeway, jobs would often exceed their resource allotments, forcing the JC to often pay for the execution of failed jobs.} \se{I think the easiest thing to do is just add this discussion about estimating resources, it addresses the reviewer comment. I also added a couple lines to the data faults in section 4.4 about the JC detecting RP resource claims. Changing to just pay the fixed max price would require dropping $\hat\pi_c$ and seems error prone.}

\subsection{\platform Actors}

\subsubsection{Resource Provider and Job Creator}
\emph{Job Creators} (JC) have jobs to outsource and are willing to pay for computational resources. \emph{Resource Providers} (RP) have available resources and are willing to let Job Creators use their hardware and electricity in exchange for monetary compensation.
JCs post offers to the market specifying the jobs, quantities of resources required, deadlines for execution, required Docker base layers, computation architecture, and unit bid prices for resources; while RPs post offers specifying available resources, Docker base layers, computation architecture, and unit ask prices for resources. To prevent the JC and RP from cheating, they are required to include security deposits (see \cref{eq:deposit}) in the offers submitted to the platform. We provide more details on the costs and deposits later in \cref{sec:analysis}. JCs and RPs also specify trusted Directories and Mediators, which we describe below. 

\revision{A JC has multiple strategies for verifying the correctness of the results returned by an RP. Any reasonable verification strategy must cost less than what it would cost to simply execute the job%
. A straightforward verification strategy is to re-execute a random subset of jobs and compare the results with those provided by the RP. As long as the number of verified jobs  is low enough, this strategy is viable. If the JC cannot execute the job itself, then it can post duplicate job offers instead and compare the results provided by different RPs. For some jobs,
there exist verification algorithms that cost significantly less than execution,  and so the JC may be able to verify every job. However, this requires the JC to implement an efficient algorithm for verification, which might be challenging.}

\subsubsection{Directory}
\emph{Directories}
are network storage services that are available to both JCs and RPs for transferring jobs and job results. 
Directories are partially trusted by the actors, which means that actors (RPs and JCs) choose to trust certain Directories, but these are not necessarily trusted by the platform or by other RPs or JCs. Directories are paid by the JC and RP for making its services available for the duration of a job.



\subsubsection{Solver}
Matching
 a JC offer, and an RP offer requires computations that cannot be executed on a smart contract which has limited computation capabilities. Therefore, we extend the concept of hybrid solvers introduced in \cite{transax} and include \emph{Solvers} in \platform. 
A Solver can be a standalone service, or it may be implemented by another actor (e.g., RPs and JCs may act as Solvers for their own offers).
Unlike the smart contract, Solvers are not running on a trustworthy platform; hence, the contract has to check the feasibility of matches that the Solvers provide which is significantly easier computationally than finding matches. Solvers receive a fixed payment, set by the platform and paid by the JC and RP, for finding a match that is accepted by the platform.

\subsubsection{Mediator}

When there is a disagreement between a JC and an RP on the correctness of a job description or a job result (e.g., the RP claims that a job result is correct, while the JC claims that it is incorrect), a  partially trusted \emph{Mediator}  decides who is at fault or ``cheating.'' A Mediator is capable of executing the job in the same way  as the RP, but it can be more expensive since it is expected to provide a more reliable service and to maintain its reputation in the ecosystem. Each Mediator sets a price which it is paid by the JC and RP for making its services available for the duration of a job. In case of mediation it is additionally compensated for the computations it executes.

\subsubsection{Smart Contract}
\label{sec:smartcontract}

The \emph{Smart Contract} (SC) is the cornerstone of our framework.
Most communication in \platform is effectuated through function calls to the SC and through events emitted~\footnote{\revision{\emph{Emitted} events in platforms such as Ethereum are recorded to the transaction logs of the ledger, which can be accessed by interested agents via polling. We use the word \emph{emit} because that is word used for this functionality in Solidity, which we use for our proof-of-concept implementation.}}  by the SC. The SC is deployed and executed on a trustworthy decentralized platform, like the Ethereum blockchain~\cite{wood2014ethereum},
\iflong
\footnote{In \platform, we use Ethereum as the distributed ledger; and for our experiments, we set up a private Ethereum network.}
\fi
which enables it to enforce the rules of the \platform protocol described in the next section.
It also enables actors to make financial deposits and to withdraw funds on conditions set by the SC. The functions provided by the smart contract can be found
\FullText{in \cref{app:SCFunc}.}{in our full paper \cite{modicum-arxiv}.}


\begin{figure*} [ht]
\centering
\subfloat[]{
 \begin{minipage}[b]{0.69\textwidth}
   \resizebox{1.\linewidth}{!}{
\begin{tikzpicture}[x=6.5cm, y=-0.063cm, semithick, font=\small,
Event/.style={blue, dashed}]
\def\posJC{0}
\def\posSolver{0.5}
\def\posPlatform{1}
\def\posDirectory{1.5}
\def\posRP{2}
\def\posM{2.4}

\foreach \name/\pos in {JobCreator/\posJC, ResourceProvider/\posRP, Smart Contract/\posPlatform, Solver/\posSolver, Directory/\posDirectory, Mediator/\posM} {
  \draw (\pos, 0) -- (\pos, 210);
  \node [draw, fill=white, rounded corners=0.1cm] at (\pos, 0) {\name};
}
\newcounter{seqTime}
\setcounter{seqTime}{0}
\foreach \action/\from/\to/\style/\offset/\delta in {
  {registerMediator/\posM/\posPlatform/solid/0/15},
  {MediatorRegistered/\posPlatform/\posJC/Event/0/7},
  {/\posPlatform/\posSolver/Event/0/0},
  {MediatorRegistered/\posPlatform/\posRP/Event/0/0},
  {mediatorAddFirstLayer/\posM/\posPlatform/solid/3ex/7},
  {MediatorAddedFirstLayer/\posPlatform/\posSolver/Event/-3ex/0},
  {registerJobCreator/\posJC/\posPlatform/solid/0/7},
  {registerResourceProvider/\posRP/\posPlatform/solid/0/0},
  {ResourceProviderRegistered/\posPlatform/\posSolver/Event/0/7},
  {JobCreatorRegistered/\posPlatform/\posSolver/Event/0/7},
  {jobCreatorAddTrustedMediator/\posJC/\posPlatform/solid/0/7},
  {resourceProviderAddTrustedMediator/\posRP/\posPlatform/solid/0/0},
  {ResourceProviderAddedTrustedMediator/\posPlatform/\posSolver/Event/0/7},
  {JobCreatorAddedTrustedMediator/\posPlatform/\posSolver/Event/0/7},
  {resourceProviderAddFirstLayer/\posRP/\posPlatform/solid/6ex/7},
  {ResourceProviderAddedFirstLayer/\posPlatform/\posSolver/Event/-6ex/0},
  {uploadJob/\posJC/\posDirectory/set1-green/0/7},
  {postJobOffer/\posJC/\posPlatform/solid/0/7},
  {postResOffer/\posRP/\posPlatform/solid/0/0},
  {ResourceOfferPosted/\posPlatform/\posSolver/Event/0/7},
  {JobOfferPosted/\posPlatform/\posSolver/Event/0/7},
  {JobOfferPosted/\posPlatform/\posRP/Event/0/0},
  {cancelJobOffer/\posJC/\posPlatform/gray/0/7},
  {cancelResOffer/\posRP/\posPlatform/gray/0/0},
  {JobOfferCanceled/\posPlatform/\posSolver/Event/0/7},
  {ResourceOfferCanceled/\posPlatform/\posSolver/Event/0/7},
  {postMatch/\posSolver/\posPlatform/solid/0/7},
  {Matched/\posPlatform/\posRP/Event/0/7},
  {Matched/\posPlatform/\posJC/Event/0/0},
  {timeout/\posJC/\posPlatform/gray/0/7},
  {getJobImage/\posRP/\posDirectory/set1-green/0/0},
  {uploadResult/\posRP/\posDirectory/set1-green/0/7},
  {postResult/\posRP/\posPlatform/solid/0/7},
  {ResultPosted/\posPlatform/\posJC/Event/0/0},
  {getResult/\posJC/\posDirectory/set1-green/0/7},
  {acceptResult/\posJC/\posPlatform/solid/0/7},
  {acceptResult/\posRP/\posPlatform/gray/0/1},
  {rejectResult/\posJC/\posPlatform/red/0/7},
  {JobAssignedForMediation/\posPlatform/\posM/Event/0/0},
  {getJobImage/\posM/\posDirectory/set1-green/0/7},
  {postMediationResult/\posM/\posPlatform/red/0/7},
  {MediationResultPosted/\posPlatform/\posJC/Event/0/0},
  {MatchClosed/\posPlatform/\posJC/Event/0/7},
  {MatchClosed/\posPlatform/\posRP/Event/0/0}%
} {
  \addtocounter{seqTime}{\delta}
  \draw [->, >=stealth, shorten <=0.05cm, shorten >=0.05cm, \style] (\from, \value{seqTime}) -- (\to, \value{seqTime}) node [midway, above, xshift=\offset, align=left, text width={abs(\to - \from) * 5.1cm}, fill=white, fill opacity=0.5, text opacity=1] {\footnotesize\texttt{\scriptsize\action}};
  \ifthenelse{\equal{\delta}{0}}{}{\draw[fill=black] (\from, \value{seqTime}) circle (1pt) node {};}
}

\draw [decorate,decoration={brace,amplitude=10pt,mirror},xshift=-2pt,yshift=0pt]
(0,10) -- (0,79) node [black,midway,xshift=-0.6cm,rotate=90] 
{\footnotesize Registration};

\draw [decorate,decoration={brace,amplitude=10pt,mirror},xshift=-2pt,yshift=0pt]
(0,81) -- (0,130) node [black,midway,xshift=-0.6cm,rotate=90] 
{\footnotesize Offers};

\draw [decorate,decoration={brace,amplitude=10pt,mirror},xshift=-2pt,yshift=0pt]
(0,131) -- (0,142) node [black,midway,xshift=-0.6cm,rotate=90] 
{\footnotesize Matching};

\draw [decorate,decoration={brace,amplitude=10pt,mirror},xshift=-2pt,yshift=0pt]
(0,143) -- (0,182) node [black,midway,xshift=-0.8cm,rotate=90, text width=2cm,align=center] 
{\footnotesize Execution and Verification};

\draw [decorate,decoration={brace,amplitude=10pt,mirror},xshift=-2pt,yshift=0pt]
(0,182) -- (0,210) node [black,midway,xshift=-0.6cm,rotate=90, text width=2cm,align=center] 
{\footnotesize  Mediation};

\end{tikzpicture}
}
\end{minipage}
\label{fig:seqdiag}
}
\subfloat[]{
 \begin{minipage}[b]{.3\textwidth}
 \includegraphics[width=1.06\linewidth, viewport=2cm 9.5cm 14.5cm 29.5cm,clip] {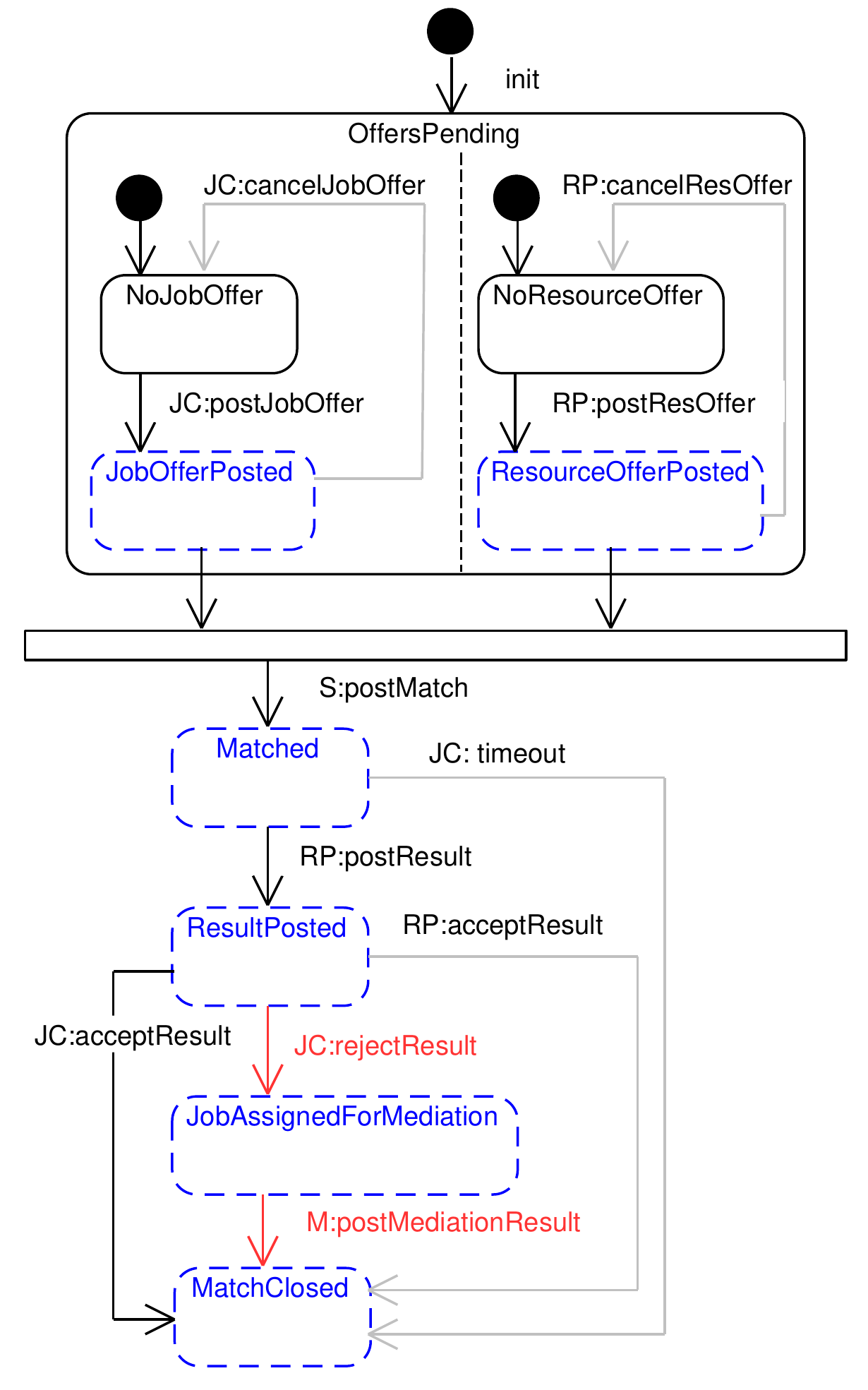}
   \end{minipage}  
   \label{fig:FSM}
   }
  \caption{(a) Sequence diagram showing the outsourcing of a single job. Black arrows are function calls to the smart contract. Blue dashed lines are events emitted by the smart contract. Gray lines are optional function calls. Red lines are optional calls that are required in case of disagreement between RP and JC. Green lines are off-chain communication. Note that events are broadcast, and visible to all agents that can interact with the contract. (b) States of job in \platform. Function calls to the smart contract are prefixed by the actors that make the calls. 
  }
  \vspace{-0.1in}
\end{figure*}
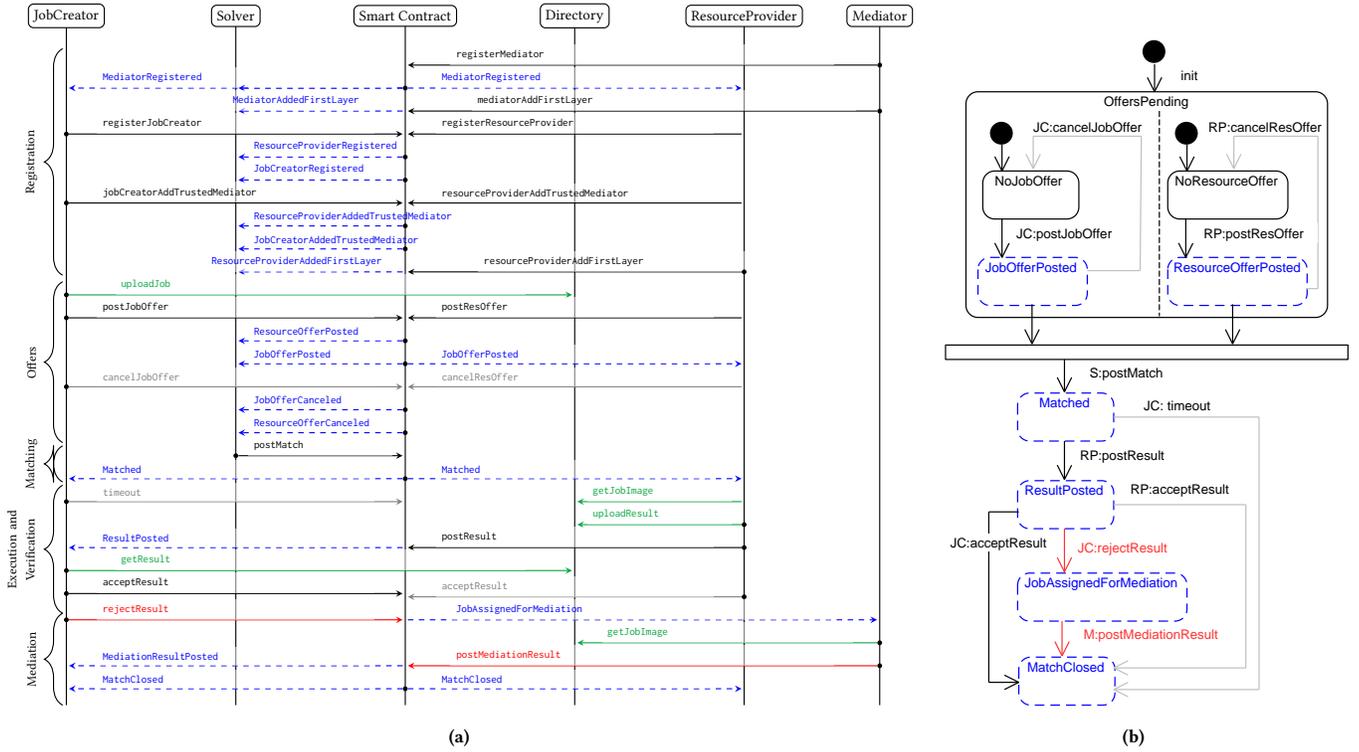

   

\section{\platform Protocol}
\label{sec:protocol}

In this section, we discuss the operation protocol, possible misbehavior, the concept of mediation, and various faults that can occur  and how they can be handled.
\cref{fig:seqdiag}  shows a possible activity sequence from registration to completion of a job. \cref{fig:FSM}  represents the state of a job from the perspective of the smart contract.

\subsection{Registration and Posting Offers}
\label{sec:offers}
First, RPs and JCs register themselves with \platform (\cref{fig:seqdiag}: \texttt{registerResourceProvider}, \texttt{registerJobCreator}). 
Note that RPs and JCs need to register only once, and then they can make any number of offers.
If an RP is interested in accepting jobs, it will send a resource offer to the platform
(\cref{fig:seqdiag,fig:FSM}:~\texttt{postResOffer}). On the other side, a JC first creates a job and uploads the job to a Directory that it trusts and can use (\cref{fig:seqdiag}: \texttt{uploadJob}). Then, it posts a job offer
(\cref{fig:seqdiag,fig:FSM}: \texttt{postJobOffer}).  

Note that any time before the offers are matched, RPs and JCs can cancel their offers. Cancellation can be due to  unscheduled maintenance, or because their offers have not been matched for a long time and they wish to adjust their offer (e.g., increase maximum price) by cancelling the previous offer (\cref{fig:seqdiag,fig:FSM}: \texttt{cancelJobOffer} and \texttt{cancelResOffer}) and posting a new one. \revision{Once matched, cancellation is no longer permitted.}

\subsection{Matching Offers}
After receiving offers, the smart contract notifies the Solvers
by emitting events (\cref{fig:seqdiag,fig:FSM}: \texttt{JobOfferPosted}, \texttt{ResourceOffer\-Posted}). The Solvers then try to match new offers to previously unmatched offers or store them so they can be matched later if no matching candidate is available. When a match is found by a Solver, the Solver posts the match consisting of the resource offer, job offer, and the Mediator to the smart contract (\cref{fig:seqdiag,fig:FSM}: \texttt{postMatch}). The smart contract checks that the submitted match is feasible, and if it is, then it records it and notifies both the JC and RP that their offers have been matched (\cref{fig:seqdiag,fig:FSM}: \texttt{Matched}). 
Note that for a match to be feasible the resources specified in the RP offer must satisfy the requirements of the job specified by the JC offer. Additionally, they should have a common architecture, trusted mediator, and directory.  
\FullText{The full feasibility specification can be found in \cref{app:feasible}.}
{The full feasibility specification can be found in our full paper \cite{modicum-arxiv}.}
With the \texttt{Matched} event, the contract pays the Solver the amount RP and JC specified as the matching incentive.

\iflong
and 
\begin{align*}
\exists\, M: ~ &M \in (JC.\textnormal{\textit{trusted\_mediators}} \cap RP.\textnormal{\textit{trusted\_mediators}}) \nonumber \\ 
& \wedge M.\textnormal{\textit{architecture}} = RP.\textnormal{\textit{architecture}} \\
& \wedge JO.\textnormal{\textit{directory}} \in M.\textnormal{\textit{trusted\_directories}} 
\end{align*}

the following conditions:
\begin{align*}
RO.\textnormal{\textit{instruction\_capacity}} &\geq JO.\textnormal{\textit{instructions\_limit}}  \\
RO.\field{ram\_capacity} &\geq JO.\field{ram\_limit} \\
RO.\textnormal{\textit{local\_storage\_capacity}} &\geq JO.\textnormal{\textit{local\_storage\_limit}}  \\
RO.\textnormal{\textit{bandwidth\_capacity}} &\geq JO.\textnormal{\textit{bandwidth\_limit}} \\
RO.\textnormal{\textit{instruction\_price}} &\leq JO.\textnormal{\textit{instruction\_max\_price}}  \\
RO.\textnormal{\textit{bandwidth\_price}} &\leq JO.\textnormal{\textit{max\_bandwidth\_price}}  \\
JO.\textnormal{\textit{architecture}} &= RP.\textnormal{\textit{architecture}} \\
JO.\textnormal{\textit{directory}} &\in RP.\textnormal{\textit{trusted\_directories}}
\end{align*}
\vspace{-2em}
\begin{align*}
\exists\, M: ~ &M \in (JC.\textnormal{\textit{trusted\_mediators}} \cap RP.\textnormal{\textit{trusted\_mediators}}) \nonumber \\ 
& \wedge M.\textnormal{\textit{architecture}} = RP.\textnormal{\textit{architecture}} \\
& \wedge JO.\textnormal{\textit{directory}} \in M.\textnormal{\textit{trusted\_directories}} 
\end{align*}
\vspace{-2em}

\begin{align*}
    &\textnormal{\textit{current\_time}}  + RP.\textnormal{\textit{time\_per\_instruction}} \cdot JO.\textnormal{\textit{instruction\_limit}} \nonumber \\ 
    & \leq JO.\field{completion\_deadline}
\end{align*}
and both the JC and RP made sufficient deposits (see \cref{sec:analysis}). 


At this point the price of the job can be calculated as follows:
\begin{align}
    & \field{result}.\field{instruction\_count} \cdot 
     \field{RO}.\field{instruction\_price} \nonumber \\
    + ~ & \field{result}.\field{bandwidth\_usage} \cdot
     \field{RO}.\field{bandwidth\_price} 
    \label{eq:job_price}
\end{align}
\fi

 \subsection{Execution by RP and Result Verification by the JC}
 After receiving notification of a match, the RP downloads the job from the Directory (\cref{fig:seqdiag}: \texttt{getJobImage}) and runs the job. While running the job, RP measures resource usage. Finally, when the job is done, it uploads results to the Directory (\cref{fig:seqdiag}: \texttt{uploadResult}) and reports the status of the job and resource usage measurements to the smart contract (\cref{fig:seqdiag,fig:FSM}: \texttt{postResult}). The status of the job is the state which the job execution finished. Some possible termination states include \texttt{Completed} or \texttt{MemoryExceeded}. 
\FullText{The full list of status codes can be found in \cref{app:result}}{The full list of status codes can be found in our full paper \cite{modicum-arxiv}}
and we will discuss some of them in \cref{sec:mediation}.

After receiving the notification that the result has been posted (\cref{fig:seqdiag,fig:FSM}: \texttt{ResultPosted}) the JC downloads the result (\cref{fig:seqdiag}: \texttt{getResult}) and decides whether to verify it or not. It then accepts, ignores, or rejects the result. If the JC accepts the result (\cref{fig:seqdiag,fig:FSM}: \texttt{acceptResult}) 
the contract returns deposits, pays the RP and Mediator, and closes the match (\cref{fig:FSM}: \texttt{MatchClosed}).

If the JC ignores the result, then after some time, the window for the JC to react closes and the RP is permitted to accept the result (\cref{fig:FSM}: \texttt{RP:acceptResult}) resulting in the match closing (\cref{fig:FSM}: \texttt{MatchClosed}). If the JC disagrees with the result, it rejects the result (\cref{fig:seqdiag,fig:FSM}: \texttt{rejectResult}) with a reason code such as \texttt{WrongResults} or \texttt{ResultNotFound}; then, mediation follows (\cref{fig:seqdiag,fig:FSM}: \texttt{JobAssignedForMediation}). 


\subsection{Faults and Mediation}
\label{sec:mediation}

There are essentially two types of faults that can occur in the system: \begin{enumerate*}
    \item Connectivity: this can occur for the JC, RP, and the Mediator when they try to communicate with the Directory or the smart contract. Note that the JC, RP, and Mediator do not talk directly (see \cref{fig:arch}). Solver connectivity faults are not a concern since it only interacts with the smart contract and the failure of a Solver implies that Solver does not submit a solution, but others may. Hence, we only worry about the connectivity to smart contract and the directory.
    \item Data (job input and results of execution): it can be malformed, return an exception when executed, not be available on the Directory, be verifiably incorrect, etc.
\end{enumerate*}{}

First, we discuss the connectivity with the smart contract.  If the RP cannot communicate with the smart contract, it may not receive a notification that it has been matched and 
it is also unable to call \texttt{postResult}. These are both addressed by the JC having a timeout. If the JC calls \texttt{timeout} (see \cref{fig:seqdiag,fig:FSM}) and the required waiting period has elapsed, then the smart contract pays the JC the estimated value of the job from the RPs deposit and returns the remainder. The timeout also addresses when the JC misses the \texttt{ResultPosted} message sent by the RP, since if the misses the message, it will attempt to timeout, which will fail because the result was posted, and  then it will fetch the result. Smart contract connectivity failure also means that the JC cannot send \texttt{acceptResult} and that the RP never receives a notification that the JC accepted. This is addressed by allowing the RP to bypass the JC and call \texttt{acceptResult} if the platform specified duration for the JC to respond has elapsed. Connectivity failure also means that the JC and the RP may not receive the \texttt{MediationResultPosted} message. In this case, they may respond by removing that Mediator from their trusted list; and if the result eventually arrives, they can re-add the Mediator. This also addresses when the Mediator does not get the mediation request or cannot post the mediation result. %
\revision{
  In this case, mediation is considered to have failed; to release the security deposits, we enforce a timeout via the smart contract. In the event of this timeout, we instead pay the RP half of the JC's job estimate and return the remainder of the deposits.
  Obviously, the Mediator does not get paid since mediation failed. 
 }


Now, we discuss connectivity with the Directory. If a JC cannot connect to the Directory to submit a job, it simply tries to upload to another Directory it trusts. If it cannot connect to retrieve a result then the JC must either pay the RP or request mediation with the \texttt{DirectoryUnavailable} status. 
To verify the JC claim, the Mediator queries the Directory for its uptime. If the Directory reports that it was available for the entire job duration, then the Mediator assigns fault to the JC; otherwise, it assigns fault to the Directory.  
If the JC does not agree, it may remove the Directory and/or the Mediator from its trusted lists. Similarly, if the RP cannot connect to the Directory, either to fetch the job or upload the results, then it posts a result with the \texttt{DirectoryUnavailable} status. If the JC does not agree that the Directory is unavailable, then it requests mediation which proceeds as above with the RP rather than the JC. If the Mediator sends a mediation result of \texttt{DirectoryUnavailable} then the RP and JC may choose to remove the Directory, Mediator, or both from their trusted list. 

Finally, we discuss the data faults. The data faults that the RP can detect and the corresponding result status are: 
1) no job on the Directory (\texttt{JobNotFound}), 2) a job description error (\texttt{JobDescription\-Error}), 3) excessive resource consumption (\texttt{ResourceExceeded}), and 4) an execution exception during execution (\texttt{Exception\-Occur\-red}). The JC can request mediation if it disagrees with a claim of any of these faults, in addition to detecting no result on the Directory, or that the result is incorrect. \revision{Finally, the JC can request mediation if, after verification, the JC suspects that the resource usage claimed by the RP is too high. This could occur since the resource variation should be small, but the JC may have set a high resource limit to ensure the job completes.} \se{@Aron: Please check this adjustment}

If the JC requests mediation claiming there is a data fault, then the Mediator attempts to replicate the steps taken by the RP, with the distinction being that it re-executes the job $n$ times (see \cref{subsec:game_model}, it is defined as a parameter of the smart contract), and compares its results with the RP's results. In two cases, the JC will be at fault:
\begin{enumerate*}
    \item All of the Mediator's results and RP's result are the same, which means that the RP has executed the job correctly.
    \item The Mediator gets two different results when running the job, which means that JC has submitted a non-deterministic job.
\end{enumerate*}
 Otherwise, the Mediator assumes that the RP has submitted a wrong result. Another case is when the JC claims that the result is not on the Directory. In that instance, the Mediator attempts to retrieve the result from the Directory. If it cannot it faults the RP, if it can it faults the JC. If either agent disagrees it may remove that Directory, Mediator, or both from its trusted list. 
 
 
 
The Mediator submits the verdict to the smart contract (\cref{fig:seqdiag,fig:FSM}: \texttt{postMediationResult}), and the  smart contract claims the security deposit. Of the deposit, \emph{the actual job price} is used to compensate the damaged party for its losses, and  $\pi_m$ (which is the job price times the number of repeated executions $n$) goes to the Mediator to cover its mediation costs. In addition, Mediators always receive $\pi_a$ as payment for making their service available. They receive this when a job is closed (\cref{fig:seqdiag,fig:FSM}: \texttt{MatchClosed}). 
 


For this mediation approach to work, the RP must not allow jobs to 
access any extra information (e.g., physical location, time) beyond what is in its description and Docker image. 
Otherwise, a job could determine where it is running 
(e.g., via connecting to a remote server), 
so the JC could create a job that would always produce different results on the RP and on the Mediator. For such a job, the Mediator would always incorrectly blame and punish the RP. This is why the platform requires that jobs be 
\begin{enumerate*}
    \item deterministic and punishes the JC if non-determinism is detected (otherwise, we could not use repeated executions to verify) and
    \item batch (otherwise, the jobs could not be isolated).
\end{enumerate*}
\camera{The deterministic restriction specifically requires that the Mediator gets the same result as the RP. This means that in practice non-deterministic jobs could be sent by the JC as long as the RP records the non-deterministic values instantiated on the RP (which are not part of the input posted on the Directory) as part of its result for use in verification. Some examples of jobs that could be computed in such environment are machine learning tasks, or jobs that are similar to volunteer computing tasks, such as protein folding. Essentially, any batch data processing task is feasible.}




\subsubsection{Collusion}
\revision{A part of the challenge in designing a fair system is the problem of collusion. We enumerate all possible two-party collusions and discuss their objectives and how \platform addresses them. We do not consider more than two party collusions explicitly because they are indistinguishable from two-party collusions  for the non-colluding agents  in \platform.  

    \begin{itemize}[noitemsep]
        \item {\it Job Creator and Solver}:   Since offers are public and any participant can act as a solver, the collusion between JCs and solvers is inevitable. The goal of this collusion is to match a JC's offers to the resource offers with the lowest unit price. This means that the RP with the lowest-price offer will be matched first. This collusion does not harm the system because every resource offer includes a minimum reservation price; 
        thus, a JC cannot force an RP to perform computation for less than what the RP voluntarily accepts. This collusion is desirable because it provides solver resilience to the system due to the incentive for JCs to form this collusion.

        \item {\it Job Creator and Mediator}:  Both JC and Mediator can benefit from this collusion by taking the RP's security deposit and splitting it between them, while the JC can also benefit by having its jobs executed without paying. 
        This can be achieved by the JC requesting mediation on a correct result and the Mediator ruling in favor of the JC. To an honest RP, this collusion will appear as a faulty Directory, faulty Mediator, or non-deterministic job, though it can eliminate the last by repeating the job execution. The 
        RP removes the 
        Mediator and Directory from its trusted list. Thus, a mediator can launch this attack  only once per RP. 
        
        \item {\it Job Creator and Directory}:  This collusion is similar to the one between the JC and Mediator. 
        JC and Directory can collude in multiple ways. For example,
        the Directory manipulates the job so that the RP will return \texttt{JobDescriptionError} result status. Then, JC will request mediation, the Directory will provide the correct job to the Mediator, and the Mediator will rule in favor of JC. 
        Since the RP cannot distinguish this collusion from the one between a JC and mediator, as a response, 
        the RP removes the 
        Mediator and Directory from its trusted list. Thus, a Directory can  launch this attack only once per RP. 

        \item {\it Job Creator and Resource Provider}:   There is no possible benefit from this collusion.

        \item {\it Resource Provider and Solver}:   The goal of this collusion is the same as for the collusion between JC and Solver, and its impact and mitigation are also the same. 
        This collusion is desirable for the same reason as well.  
      

        \item {\it Resource Provider and Mediator:} RP and Mediator can both benefit from this collusion by taking the JC's security deposit and splitting it between them; while the RP can also benefit by receiving payment for a job without actually executing it. 
        This can be achieved by the RP returning any job result, which the JC might verify and request to be mediated, upon which Mediator will rule in favor of the RP. 
        This collusion is mitigated in the same way as the collusion between JC and Mediator, except that the roles of JC and RP are reversed.

        \item {\it Resource Provider and Directory}:    The goal of this collusion is the same as the RP and Mediator collusion, and it can be achieved and handled in a similar manner as the JC and Directory collusion. To an honest JC, collusion will appear as a faulty or colluding Directory. As a response, the
        JC removes the Mediator and Directory from its trusted list. Thus, a Directory can  launch this attack only once per JC.  

        \item {\it Directory and Solver}:   There is no benefit from this collusion.

        \item {\it Directory and Mediator}:   This collusion aims to ensure that the JC will request Mediation by manipulating  data or availability, and then splitting the payment for Mediation. Depending on the Directory's manipulation and the Mediator's ruling, either the JC or RP will respond in the same way as if the RP or JC were colluding with the Directory or Mediator.

        \item {\it Mediator and Solver}: The goal in this case is for the solver to prioritize trades that include the  Mediator, and further prioritize JCs and RPs with a history of requiring mediation. This could result in unfairness, i.e., some jobs may never get matched. However, this is not a real concern since the JC and RP can act as solvers and match their own offers. 
    \end{itemize}{}

    \reviewercomment{it is possible that the Mediators (and Directory) acts dishonestly.}
    From exploring the possible scenarios, we conclude that the JC and RP could be cheated once by a Mediator or a Directory; however, the faulty agent will be removed from the trusted list afterwards. Since the business model for Mediators and Directories  is attracting RPs and JCs who trust and pay them for their services, they have strong incentives to build a positive reputation in the ecosystem. Thus, we assume they will behave honestly. The formal proof for this conjecture will be provided in future work. 
    }

\begin{table}[t]
    \centering
        \caption{\platform Parameters and System Constraints}
    \label{tab:symbols}
    \vspace{-0.5em}
    \resizebox{\columnwidth}{!}{%
    
\begin{tabulary}{\columnwidth}{|p{0.05\columnwidth}|J|}
    \hline
    \multicolumn{2}{|c|}{\bf \platform Smart Contract (SC)}\\
    \hline
    $\theta$ & penalty rate set by the contract \\
    $d$ & deposit by JC and RP for collateral before transaction occurs  \\
    $d_{min}$ & minimum security deposit \\
    $g_{m}$ & cost of requesting mediation\\
    $g_{r}$ & cost for an RP to participate; includes the costs of submitting the offer and the results, as well as partial payment to the Solver for a match accepted by the smart contract \\
    $g_{j}$ & cost for a JC to participate; includes the costs of submitting the job offer, as well as partial payment to the Solver for a match accepted by the smart contract \\
    $\pi_d$  & payout to wronged party when \emph{deception} is detected \\
    $n$  & number of times Mediator executes a disputed job \\
    \hline
    \multicolumn{2}{|c|}{\textbf{Mediator (M) \ifSCDIR{and Directory (D)}\fi}}\\
    \hline
    $\pi_m$  & payout to the \emph{Mediator} when it is invoked \\
    $\pi_{a}$  & payout to the Mediator \ifSCDIR{and Directory }\fi for being \emph{available} for the duration of the job, which also covers the Solver match payment\\
    \hline
    \multicolumn{2}{|c|}{\textbf{Job Creator (JC)}} \\
    \hline
    $p_v$ & probability that JC verifies a job result (verification rate) \\
    $p_a$ & probability that a correct execution of a non-deterministic job returns a ``normal'' result for which the JC will not request mediation (functionally, this is an indicator of how honest the JC is)\\ 
    $\pi_c$ & payment from JC to RP for successfully \emph{completing} a computation\\
    $b$ & JC's benefit for finished job minus cost of submitting job \\
    $c_{v}$ & JC's cost to verify a job result \\
    \hline
    \multicolumn{2}{|c|}{\textbf{Resource Provider (RP)}} \\
    \hline
    $\pi_r$ & payout that \emph{resource} provider asks for completing the job \\
    $p_e$ & probability that RP intentionally \emph{executes} the job correctly ($1-p_e$ is the probability of cheating) \\
    $c_e$ & \emph{execution} cost for RP to compute job \\
    $c_d$ & computational cost for RP to \emph{deceive} and create wrong answer  \\
    \hline
    \multicolumn{2}{|c|}{\textbf{System Constraints}} \\
    \hline
    1 & $b>\pi_{c}+\pi_{a}+g_{j}$ for honest JC \\
    2 & $\theta \geq0$ the penalty rate cannot be negative \\
    3 & $n>0$ else the mediator does not re-execute the job \\
    4 & $c_e > c_d > 0$ else the RP never has incentive to cheat \\
    5 & $\hat\pi_{c} \geq \pi_{c} \geq \pi_{r} > c_e$; if $c_e \geq \pi_{r}$, the RP will abort the job; and if $\pi_{r} > \hat\pi_{c}$, then that job and resource offer match is disallowed by the contract \\
    6 & $d \geq d_{min}$ \\
    \hline
\end{tabulary}
    }

\end{table}

\section{Analyzing Participant Behavior and Utilities}
\label{sec:analysis}



In this section, we formulate the agents' actions as a game and solve for strategies that result in a Nash equilibrium. 
We also show that platform parameters can be set so that a rational JC will follow the protocol with, at least, some minimum probability. The extensive-form representation of the game (shown in \cref{fig:gameTree}) is explained~below.  

\subsection{Game-Theoretic Model}
\label{subsec:game_model}

To understand the game, we first introduce a set of parameters in \cref{tab:symbols}. Parameters denoted by $p$ are probabilities, $\pi$ are payouts, $c$ are costs incurred by agents, and $g$ are the costs for interacting with the platform. Many parameters have constraints on the valid values that they can take and on their relationships with other parameters. 

Now, we consider the JC's choices. The JC has a job to outsource, whose execution provides a constant benefit $b$ upon receiving the correct job result. The JC is willing to pay a job specific price, up to $\hat\pi_{c}$, which is appropriate for the resources required to have the result computed. The JC is able to verify if the job result is correct, but it incurs verification cost $c_v$ by doing so. In order to mitigate this cost, the JC may choose to verify the result with some probability $p_v$, trading confidence for lower costs. 




\begin{figure}[t]
    \centering
    \FullText
        {\includegraphics[width=\columnwidth]{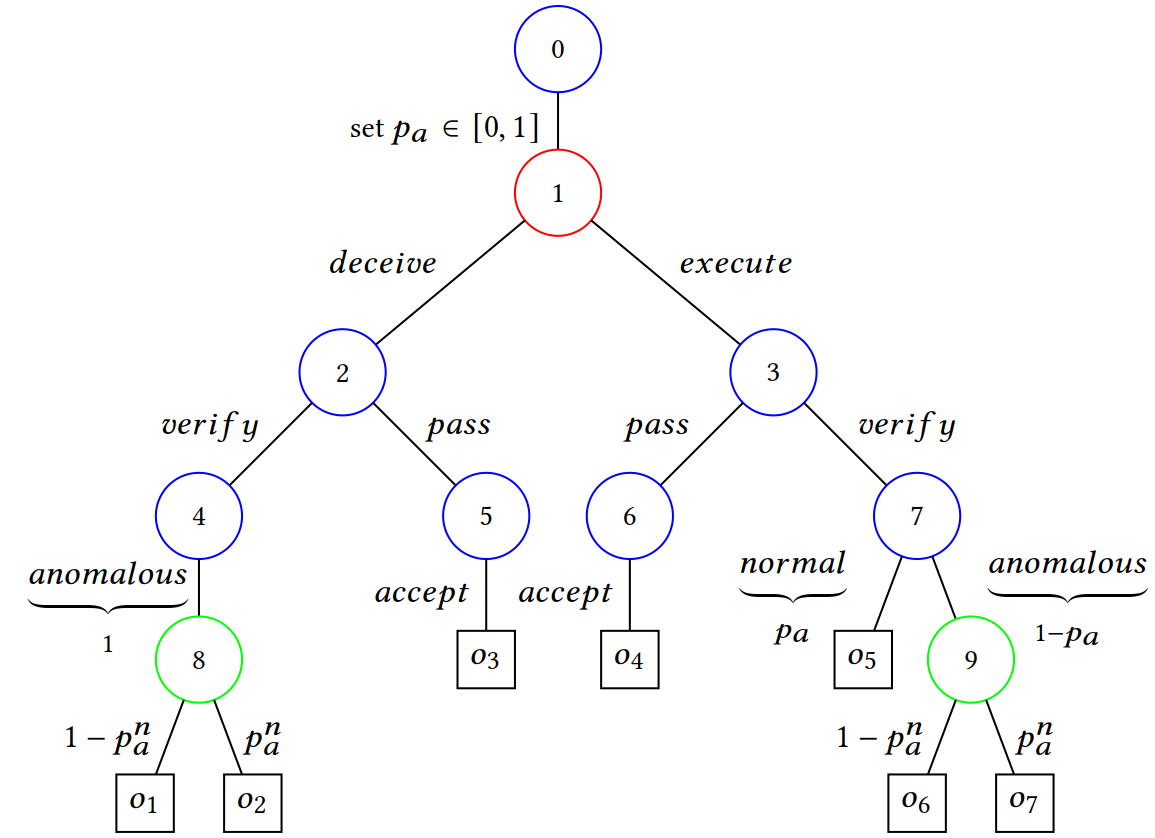}}
        {\input{figures/gameTree.tex}
    }
    \caption{Extensive-form game produced by the \platform protocol. Blue nodes indicate JC moves, red nodes indicate RP moves, green nodes indicate Mediator's probabilistic outcome. The game is sequential, but the decisions are hidden, so we treat it as a simultaneous move game. Each outcome of the game has payouts for the agents, which are found in \cref{fig:payoffs}}
    \label{fig:gameTree}
    \vspace{0.4em}
    
    \resizebox{\columnwidth}{!}{\begin{tabular}{@{}cllll@{}}
\toprule
\multicolumn{1}{l}{Outcome} & Party & \begin{tabular}[c]{@{}l@{}}Contract\\ Payoff\end{tabular} & \begin{tabular}[c]{@{}l@{}}Self \\ Benefit\end{tabular} & Reward(r)         \\ \hline
\multirow{2}{*}{$o_1$}      & RP    & $\pi_{d}-g_{r}-\pi_{a}$                                         & $-c_d$                                              & $\pi_{d} - c_d-g_{r}-\pi_{a}$     \\
                            & JC    & $-g_{j}-d-g_{m}-\pi_{a}$                                                & $-c_v$                                      & $-g_{j}-d-g_{m}-c_v-\pi_{a}$  \\
                            \hline
\multirow{2}{*}{$o_2$}      & RP    & $-d-g_{r}-\pi_{a}$                                                      & $-c_d$                                      & $-d-g_{r}-c_d-\pi_{a}$      \\
                            & JC    & $\pi_{d}-g_{j}-g_{m}-\pi_{a}$                                           & $-c_v$                                      & $\pi_{d}-g_{j}-g_{m}- c_v-\pi_{a}$  \\
                            \hline
\multirow{2}{*}{$o_3$}      & RP    & $\pi_{c}-g_{r}-\pi_{a}$                                                 & $-c_d$                                      & $\pi_{c}-g_{r}- c_d-\pi_{a}$   \\
                            & JC    & $-g_{j}-\pi_{c}-\pi_{a}$                                                & 0                                           & $-g_{j}-\pi_{c}-\pi_{a}$        \\
                            \hline
\multirow{2}{*}{$o_4$}      & RP    & $\pi_{c}-g_{r}-\pi_{a}$                                                 & $-c_e$                                      & $\pi_{c}-g_{r}- c_e-\pi_{a}$   \\
                            & JC    & $-g_{j}-\pi_{c}-\pi_{a}$                                                & $b$                                         & $b-g_{j}-\pi_{c}-\pi_{a}$       \\
                            \hline
\multirow{2}{*}{$o_5$}      & RP    & $\pi_{c}-g_{r}-\pi_{a}$                                                 & $-c_e$                                      & $\pi_{c}-g_{r}-c_e-\pi_{a}$   \\
                            & JC    & $-g_{j}-\pi_{c}-\pi_{a}$                                                & $b-c_v$                                     & $b-g_{j}-\pi_{c}-c_v-\pi_{a}$   \\
                            \hline
\multirow{2}{*}{$o_6$}      & RP    & $\pi_{d}-g_{r}-\pi_{a}$                                                 & $-c_e$                                      & $\pi_{d}-g_{r}- c_e-\pi_{a}$     \\
                            & JC    & $-g_{j}-d-g_{m}-\pi_{a}$                                                & $b-c_v$                                     & $b-g_{j}-d-g_{m}-c_v-\pi_{a}$ \\
                            \hline
\multirow{2}{*}{$o_7$}      & RP    & $-d-g_{r}-\pi_{a}$                                                      & $-c_e$                                      & $-d-g_{r}-c_e-\pi_{a}$      \\
                            & JC    & $\pi_{d}-g_{j}-g_{m}-\pi_{a}$                                           & $b-c_v$                                     & $b+\pi_{d}-g_{j}-g_{m}-c_v-\pi_{a}$ \\ 
\bottomrule
\end{tabular}}
    \caption{Game outcomes and payments in \cref{fig:gameTree}. For example, $o_1$ is the outcome when the RP deceives, the JC verifies, and the Mediator finds non-determinism and faults the JC. \revision{ The reward for the agents is the sum of contract payoff and self-benefit and is denoted $r_{o_i}$. The Mediator is not included in the table: in every outcome, it receives $\pi_{a}$;  when the JC requests mediation, the M also receives $\pi_{m}$, which is drawn from the faulty party's deposit $d$.}
    }
    \label{fig:payoffs}
    \vspace{-.6cm}
\end{figure}{}

However, a dishonest JC can design non-deterministic jobs and we assume that the JC can always recover the correct result from any output\footnote{Note that the RP and M cannot be expected to analyze the code, and hence cannot know that the job contains non-determinism.}. The goal of the JC in designing a non-deterministic job is to get the correct output without having to pay the RP. It can accomplish this if it requests mediation, and when it does, the mediator concludes that the result returned by the RP is incorrect. Thus, if a JC designs a job to look ``normal'' to the mediator with probability $p_{a}$ and ``incorrect'' with probability  $1-p_{a}$, then the JC will accept correct results with probability $p_{a}$ and request mediation with probability $1-p_{a}$. A simple illustration of such a job is one which returns a natural number as its solution, and changes the sign of that value (i.e., multiply by -1) with a fixed probability, creating a set of positive results and a set of negative results.

The game starts with the JC choosing a probability value for~$p_{a}$. A probability of $p_a = 0$ means that the JC is completely dishonest, and all results will be considered incorrect. A probability of $p_a = 1$, means the JC is honest and all correct outputs are accepted.  

The RP makes the next move, choosing between honestly executing the job or forging a result to deceive the JC. The RP may be motivated to return a false result because the job execution cost $c_e$ is higher than the deception cost $c_d$ and it is possible that the JC does not verify the result and thus does not detect the deception. The RP executes the job with probability $p_e$, where $p_{e}=0$ means that the RP is completely dishonest and  always attempts to deceive the JC, while $p_{e}=1$ means that the RP is honest and  executes the job correctly. Note that the correct result can be a fault code if the computation fails. 


The JC then makes the next move and selects its verification strategy,  choosing between verifying the job result or passing on the verification. The JC verifies the result with probability~$p_v$. If verification finds the result to be incorrect, the JC  requests mediation to dispute the result.

To resolve the dispute, the Mediator must determine which agent is at fault. The Mediator does this by performing the steps that an RP would take to execute a job, repeating several times to detect non-determinism. When initialized the smart contract specifies a \emph{verification count} $n$, which is the number of times the Mediator will execute a job checking for anomalies. Since the job has probability $p_a$ of returning a normal result and the Mediator executes the job $n$ times, the probability that the job returns a normal result in every execution is $p_a^n$. Thus, the Mediator detects a non-deterministic job with probability $1-p_{a}^{n}$, and fines the JC for being~dishonest.


As stated in \cref{sec:offers}, to deter cheating through fines, we require JCs and RPs to provide a \emph{security deposit} $d$ when submitting offers. We define the deposit to be dependent on the JC's estimate of the job price $\hat{\pi}_{c}$ and scaled by a penalty rate $\theta$, which is set by the smart contract. The job price $\hat\pi_{c}$ estimated by the JC is the same as $\pi_{c}$ except it uses the JC's bid prices and requested resources. The penalty rate must be set to a sufficiently high value to deter misbehavior. The security deposit must also cover the cost of potential mediation $\pi_{m}$, which we estimate as $\hat\pi_{c} \cdot n$ since the JC is willing to pay $\hat\pi_{c}$ and the Mediator must run $n$ times. 
\ifSCDIR{The deposit must also cover the availability costs of the Mediator and \Aron{Directory has no connection to the smart contract or blockchain, so its payment should not be a part of the deposit. JC can pay the Directory off-chain.}Directory as well as the Solver costs; we let $\pi_a$ denote the sum of these costs. Thus, we define the minimum security deposit $d_{min}$ as: }
\else{
The deposit must also cover the availability costs of the Mediator as well as the Solver; we let $\pi_a$ denote the sum of these costs. Thus, we define the minimum security deposit $d_{min}$ as: }
\fi
\begin{align}
\vspace{-0.1in}
    d_{min}   &= \underbrace{\hat\pi_{c} \cdot \theta}_{\text{penalty}} + \underbrace{(\hat\pi_{c}\cdot n)}_{\pi_{m}} + \pi_{a}
    \label{eq:deposit}
\end{align}


%

The game induced by the interactions of the actors described above has 7 possible outcomes. Each  outcome has payouts for the agents as described in \cref{fig:payoffs}. To illustrate how the payouts are calculated, consider the following sequence. The JC pays $g_j$ to submit a non-deterministic job with probability $p_a$ of returning a normal result. The RP honestly executes the job incurring cost $c_e$ and pays $g_r$ to submit the result. Since the RP executed honestly, the JC receives the benefit $b$. The JC verifies the result, incurring cost $c_v$, and detects that the non-beneficial part of the result is anomalous. It then attempts to avoid paying $\pi_{c}$ to the RP by requesting mediation, paying $g_{m}$. The Mediator executes the job $n$ times and if in one or more of those executions it encounters an anomalous result, which occurs with probability $1-p_{a}^n$, then it submits to the smart contract that the JC is at fault and the JC loses its security deposit $d$ for submitting non-deterministic jobs resulting in outcome $o_{6}$. Otherwise if all of the results from the repeated executions are normal then the JC successfully cheats and receives $\pi_{d}$ as reparations for being ``faulted'' resulting in $o_{7}$. The payouts of the other outcomes are calculated similarly. The platform interaction costs are fixed properties of a given smart contract and its underlying platform.
%
%

\begin{table}[t]
\caption{RP payoffs by decision}
\vspace{-0.5em}
\centering
\footnotesize
\label{tab:RP}
\resizebox{0.9\columnwidth}{!}{%
\begin{tabular}{@{}lp{5cm}l@{}}
\toprule
        & verify                                                                            & pass     \\ \midrule
execute & $\overbrace{- c_{e} - g_{r} - \pi_{a} +}^{U_{EV}^{RP}} \newline 
\pi_{c} \Big(n p_{a}^{n} \left(p_{a} - 1\right) + p_{a} + p_{a}^{n} \theta \left(p_{a} - 1\right)\Big) +  \newline
\pi_{d} \left(1-p_{a}\right)\left(1-p_{a}^{n}\right)$ & $\overbrace{\pi_{c}- c_{e} - g_{r} - \pi_{a}}^{U_{EP}^{RP}}$ \\
\rule{0pt}{15pt} 
deceive & $- c_{d} - g_{r} - \pi_{a} + \newline \underbrace{p_{a}^{n} \pi_{c} \left(- n - \theta\right) + \pi_{d} \left(1 - p_{a}^{n}\right)}_{U_{DV}^{RP}}$                     & $\underbrace{\pi_{c}- c_{d}- g_{r} - \pi_{a}}_{U_{DP}^{RP}}$ \\ \bottomrule
\end{tabular}
}
\end{table}

%
\revision{
Since we assume that Directories and Mediators  always act honestly, we do not consider their strategic decisions in our game analysis.}
The expected utilities of both the RP and the JC are summarized in \cref{tab:RP,tab:JC}, respectively. The table is constructed by considering the possible actions of the RP and JC. There are two possible actions for RP (execute and deceive) and two for JC (verify and pass) as illustrated by the tree in \cref{fig:gameTree}. Hence, the utilities of RP and JC depend on the four action combinations and their possible outcomes $o_1 \cdots o_7$. To understand how the utilities are calculated, consider the example of the case when RP chooses to execute, and JC chooses to verify. This is node 7 in \cref{fig:gameTree}, and there are three possible outcomes $o_5,o_6,o_7$. Outcome $o_5$ occurs with probability $p_a$, $o_6$ with probability $(1-p_a)(1-p_{a}^{n})$, and $o_7$ with probability   $(1-p_a)p_{a}^{n}$. Thus,
\begin{equation}
    U_{EV}^{JC} = p_{a}\cdot t_{o_5} + (1-p_a)\big((1-p_{a}^{n})\cdot t_{o_6} + p_{a}^{n}\cdot t_{o_7}\big)
    \label{eq:UN}
\end{equation}{}
The utility for each combination of actions is denoted by utility $U$ with superscript of the agent (i.e., RP and JC) and subscript of the action combination of RP and JC (i.e., EV is <execute,verify>, DP is <deceive,pass>, EP is <execute,pass>, and DV is <deceive,verify>). Note that we replace the total outcomes payoffs using \cref{fig:payoffs} in \cref{tab:RP,tab:JC}.

\subsection{Equilibrium Analysis}
\label{sec:Eqanalysis}

Here we analyze the utility functions explained in previous section. \FullText{}{For lack of space, we describe only the key results.} Detailed proofs for these statements can be found in
\FullText{\cref{app:proofs}.}{our full paper \cite{modicum-arxiv}.}

The ideal operating conditions for the platform would be if the RP always \emph{executed} ($p_e=1$), the JC never had to \emph{verify} ($p_v=0$) and only submitted jobs which returned deterministic results ($p_a=1$). However, if the agents are rational, these parameters do not constitute a Nash equilibrium. This is because if the JC does not verify, then the RP will choose to deceive rather than execute since $U_{EP}^{RP}<U_{DP}^{RP}$. The JCs utility in this case is always negative and so it is better off not participating in the platform. This proves the following theorem: 
\begin{theorem}[JC should not always pass]
If the JC always \emph{passes} (i.e., $p_v=0$), then the RP's best response is to always \emph{deceive} (i.e., $p_e=0$). 
\end{theorem}

\begin{table}[t]
\caption{JC payoffs by decision}
\vspace{-0.5em}
\footnotesize
\label{tab:JC}
\resizebox{\columnwidth}{!}{%
\begin{tabular}{@{}lp{6cm}l@{}}
\toprule
         & verify                                                                                                               & pass        \\ \midrule
execute  & $\overbrace{b - g_{j} - \pi_{c} \left(n + \theta\right) \left(1-p_{a}\right)\left(1-p_{a}^n\right) + }^{U_{EV}^{JC}} \newline \left(1 - p_{a}\right) \left(- g_{m} + p_{a}^{n} \pi_{d}\right) - c_{v} - p_{a} \pi_{c} - \pi_{a}  $ & $\overbrace{b - g_{j} - \pi_{a} - \pi_{c}}^{U_{EP}^{JC}}$   \\
\rule{0pt}{15pt} 
deceive  & $- c_{v} - g_{j} - g_{m} + p_{a}^{n} \pi_{d} - \pi_{a} \newline \underbrace{- \pi_{c} \left(n - p_{a}^{n} \left(n + \theta\right) + \theta\right)}_{U_{DV}^{JC}}$                    & $\underbrace{- g_{j} - \pi_{a} - \pi_{c}}_{U_{DP}^{JC}}$     \\  
\bottomrule
\end{tabular}
}
\end{table}


If the RP always chose to deceive, the platform would serve no purpose. Therefore, we must ensure that if the JC chooses to verify, the RP prefers to execute. This occurs when $U_{DV}^{RP} < U_{EV}^{RP}$. We show in 
\FullText{\cref{app:proofs}}{\cite{modicum-arxiv}}
that this is true if $p_a^{n+1} > \frac{1}{2}$. When these conditions are true, we can prove the following theorem:
\begin{theorem}[$p_e>0$]If $p_v>0$ and $p_{a}^{n+1}>\frac{1}{2}$, then a rational RP must execute the jobs with non-zero probability. 
\label{thm:pe}
\end{theorem}

Recall $p_a$ is a parameter set by the JC, so to satisfy the condition on $p_a$ in \cref{thm:pe} we must show that the platform can set parameters to force the JC to choose a value for $p_a$ that is greater than some lower bound. We assume that the JC is rational and chooses $p_a$ to optimize its utility $U^{JC}$. To find the bound we take its derivative with respect to $p_a$, $\frac{\partial U^{JC}}{\partial p_a}$, and assess how each parameter shifts the optimal value for $p_a$. The trends are as $n$, $\theta$, $g_m$ increase $p_a$ also increases, meanwhile as $\pi_c$ and $p_e$ increase $p_a$ decreases. Knowing how varying each parameter shifts the optimal value for $p_a$ we can select the worst-case values for each parameter, i.e. those that minimize optimal $p_a$, maximizing dishonesty.  Specifically, if the parameter and $p_a$ are inversely related, set the parameter to its maximum allowed value, and if they are directly related set the parameter to its minimum value. Thus, the worst-case values for each of the parameters are: $g_{m}=0$, $p_{e}=1$, $n=1$, $\theta=0$. The plot in \cref{fig:dJUvsPa1} uses those parameters and shows that increasing $n$ does cause the optimal value for $p_a$ to increase. We see that when $n=1$, the optimal $p_a=0.5$; and when $n=4$, the optimal $p_a=0.943$. This can be summarized as: 
\begin{theorem}[Bounded $p_a$]
    Setting the JC utility function parameters to minimize the optimal value for $p_{a}$ (maximizing a rational JC's dishonesty) ensures that any deviation will increase $p_a$. The platform  controls $n$ and $\theta$, and so controls the minimum optimal value of $p_a$.  
    \label{thm:bpa}
\end{theorem}{}
We want to minimize the number of times the Mediator has to replicate the computation, so we set $n=2$ and set $\theta=50$ which yields a minimum $p_a=0.99$.

\iflong

\begin{table}[]
\centering
\caption{How the optimal value for $p_a$ varies with an increase in each of the parameters.}
\label{tab:dJUvsPa}
\begin{tabular}{@{}ll@{}}
\toprule
parameter  action   & $p_a$                \\ \midrule
$n$        increase & increase             \\
$\theta$   increase & increase             \\
$g_m$      increase & increase             \\
$\pi_{c}$  increase & decrease up to limit \\
$p_e$      increase & decrease             \\ \bottomrule
\end{tabular}
\end{table}
\fi

\begin{figure}[tb]
 \centering
  \resizebox{.7\columnwidth}{!}
  {
    \input{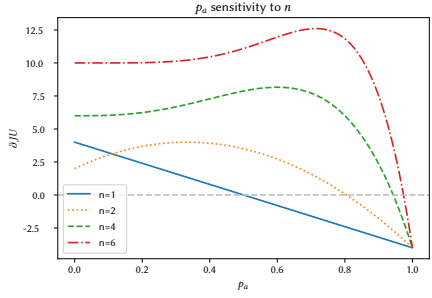}
  }
  \vspace{-1em}
 \caption{We vary the value of $n$ and plot $p_a$ against $\frac{\partial U^{JC}}{\partial p_a}$. This shows that as $n$ increases, so does the optimal value of $p_a$ (zero-crossing of derivative curve). Parameter values are $\pi_{c}=2$, $g_{m}=0$, $\theta=0$, $c_v=1$, $b=4$, $p_e=1$}
 \label{fig:dJUvsPa1}
\end{figure}


So far we have shown that we can ensure that a rational RP will prefer to execute when a JC verifies, and deceive when the JC passes, and we can limit the amount of cheating the JC can achieve through non-deterministic jobs. Next, we analyze the JC utilities to determine the Nash equilibria of the system.

\textbf{Analyzing JC types:} The preferences of the JC depend on the parameters in its utility function. We refer to each combination of preferences as a ``type'' of JC. We will call a JC that prefers to always verify as type 1. A JC that prefers to always pass is type 4; we have already covered this type and determined that a JC of this type will not participate. A JC that prefers to verify when the RP executes and pass when the RP deceives is type 2. A JC that prefers to pass when the RP executes and verify when the RP deceives is type 3. The JC has a preference because its utility is better in that case. These preferences are summarized in  \cref{tab:JCredux} with the $*$ symbol followed by the type that prefers that choice. This table is \cref{tab:JC} refactored to remove terms that do not impact the JC's preference and to highlight the relationship between the preference and the cost of verification~$c_v$. We consider the equilibrium for each type assuming the RP has been restricted as we discussed previously. The theorems below summarize these observations. Proofs are available in~\FullText{\cref{app:proofs}.}{our full paper \cite{modicum-arxiv}.}
\begin{theorem}[JC type 1]  If the JC is type 1, it will always verify ($p_v=1$) since $c_v$ is sufficiently low. This results in a pure strategy equilibrium <execute,verify>.
\end{theorem}
\begin{theorem}[JC type 2]  If the JC is type 2, it results in two pure strategy equilibria <execute,verify>, <deceive,pass>, and one mixed strategy Nash equilibrium where the JC randomly mixes between verifying and passing.
\end{theorem}
\begin{theorem}[JC type 3]  If the JC is type 3, it will result in a Nash equilibrium where the JC randomly mixes between verifying and passing, and the RP mixes between executing and deceiving.
\end{theorem}

\textbf{Strategies:} Based on these preferences, a type 1 JC will always verify $p_v=1$. Type 2 JCs may also choose to always verify, or choose a mixed strategy, setting  $p_v$ such that the RP receives the same utility regardless of whether it executes or deceives resulting in a Nash equilibrium. It achieves this by solving \cref{eq:msne_sv2} for $p_v$, setting $c_e=\pi_{c}$ and $c_d=0$. Type 3 JCs only have the option of solving \cref{eq:msne_sv2} for $p_v$. 
\begin{equation}
    \begin{aligned}{}
        p_{v}\cdot U_{EV}^{RP} + (1-p_{v})\cdot U_{EP}^{RP} &= p_{v}\cdot U_{DV}^{RP} + (1-p_{v})\cdot U_{DP}^{RP} \\
        \text{Solve for $p_v$; } \quad p_v &= \frac{c_{e}- c_{d}}{p_{a}^{n+1} \pi_{c} \left(\theta + n + 1\right)}
    \label{eq:msne_sv2}
    \end{aligned}{}
\end{equation}{}

The RP's strategy changes depending on which type of JC it is working with. If the JC is type 1, it is simple: the RP must execute. However, the other two types can mix, so the RP must also mix. It does this by solving \cref{eq:msne_se1} for $p_e$. The challenge with this is that the RP does not know the value of $c_v$. However all other parameters are known once a match is made except $p_a$ which from our work earlier we know that $p_a\geq .99$. Thus, the RP can sample $c_v$ from a uniform distribution where $0 \leq c_v \leq P_{EV}^{JC}$ for type 2 and $0 \leq c_v \leq P_{DV}^{JC}$ for type 3 ($P_{EV}^{JC}$ is the JCs preference value for <execute,verify> from \cref{tab:JCredux}). However, since the RP does not know which type of JC it is working with, it further mixes between the 3 strategies according to its belief on the distribution of the types of JC in the system. 
\begin{equation}
    \begin{aligned}{}
        p_{e}\cdot U_{EV}^{JC} + (1-p_{e})&\cdot U_{DV}^{JC} = p_{e}\cdot U_{EP}^{JC} + (1-p_{e})\cdot U_{DP}^{JC} \\
        \text{Solve for $p_e$; } \quad p_e = &\frac{2 \pi_{c} - c_{v} - g_{m} - \pi_{c} \left(1 - p_{a}^{n}\right) \left(n + \theta + 1\right) }{p_{a} \left(2 \pi_{c}- g_{m} - \pi_{c} \left(1 -p_{a}^{n} \right) \left(n + \theta + 1\right) \right)}
    \label{eq:msne_se1}
    \end{aligned}{}
\end{equation}{}


\begin{table}[t]
\caption{Simplified JC payoffs to assess dominant strategy with $\pi_d=\pi_c$}
\vspace{-0.5em}
\footnotesize
\label{tab:JCredux}
\centering
\resizebox{0.8\columnwidth}{!}{%
\begin{tabular}{p{0.9cm}p{4.75cm}r}
\toprule
        & verify                                                                                     & pass \\ \midrule
execute & $- \pi_{c} \left(1 - p_{a}\right) \left(1 - p_{a}^{n}\right) \left(n + \theta + 1\right) + \left(1 - p_{a}\right) \left(- g_{m} + 2 \pi_{c}\right) ^{*1,2} $ & $c_{v}^{*3,4}$   \\
deceive & $2 \pi_{c} - g_{m} - \pi_{c} \left(1 - p_{a}^{n}\right) \left(n + \theta + 1\right)^{*1,3}$                                                                & $c_{v}^{*2,4}$     \\ \bottomrule
\end{tabular}
}
\end{table}

This analysis shows that we can limit the dishonesty of the JC, and that the JC and RP can compute strategies that will result in a mixed-strategy Nash equilibrium. \revision{Further, we can show that the JC will not have to verify frequently by examining \cref{eq:msne_sv2} and showing that the maximum verification rate is low. First, we know that if $c_{e}<\pi_{c}$ (which should hold since otherwise the RP will always deceive), then $\frac{c_{e}-c_{d}}{\pi_{c}}\leq 1$. In this case, the verification rate is at its maximum when $c_{e}=\pi_c$ and $c_{d}=0$. Simplifying \cref{eq:msne_sv2}, we find that $\frac{1}{p_{a}^{n+1}(\theta+n+1)}$. In establishing \cref{thm:bpa}, we showed that setting $\theta$ and $n$ enforces a minimum value for $p_{a}$. Again, choosing $n=2$ and $\theta=50$ to minimize the number of times the Mediator has to replicate the computation and recalling that this results in $p_a\geq0.99$, we substitute these values into our simplified equation and find that  $p_v = 0.02$. This means that the JC will verify 2\% of the results. Since $p_a\geq 0.99$, mediation will occur 0.02\% of the time if the JC is cheating. 

}
\se{@AD: haven't found a paper that provides a bug/error frequency estimate. Not sure what I'm looking for exactly.}

\section{\platform Implementation}
\label{sec:implementation}
\begin{figure*}{}
\centering
\subfloat[Job Creator CPU]{\includegraphics[width=.33\textwidth]{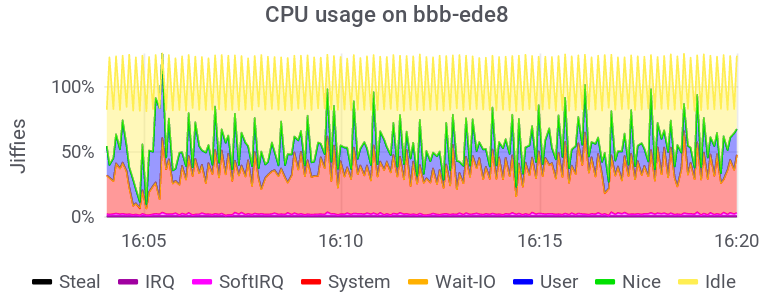}}\hfil \subfloat[Resource Provider CPU] {\includegraphics[width=.33\textwidth] {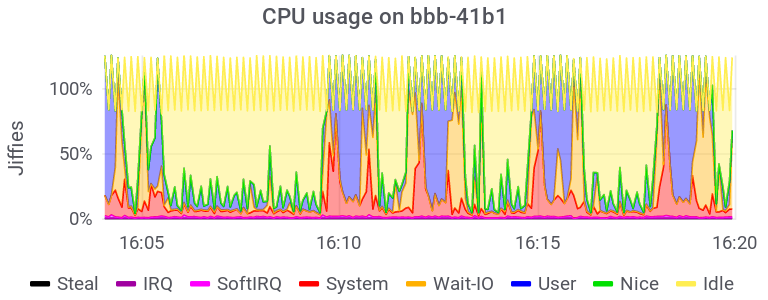}\label{fig:RPCPU_OH}}
\hfil \subfloat[Mediator CPU]{\includegraphics[width=.33\textwidth]{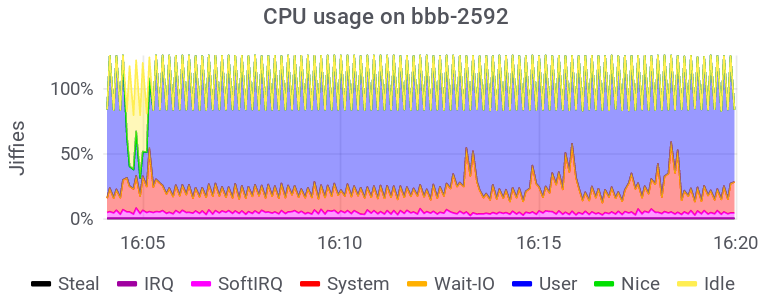}} 

\subfloat[Job Creator RAM]{\includegraphics[width=.33\textwidth]{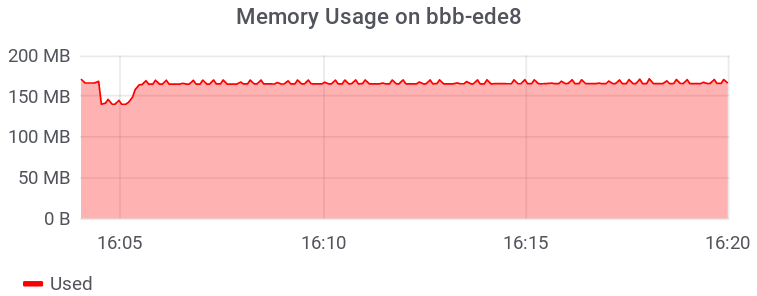}}\hfil   
\subfloat[Resource Provider RAM]{\includegraphics[width=.33\textwidth]{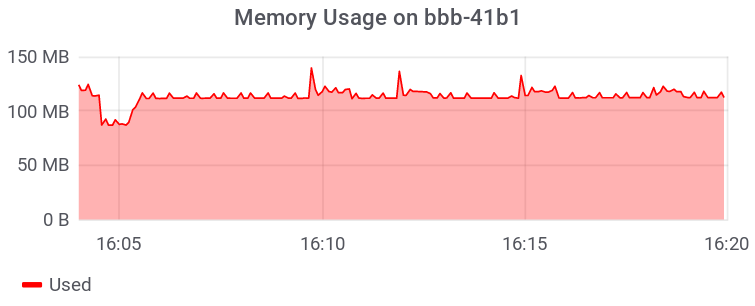}}\hfil
\subfloat[Mediator RAM]{\includegraphics[width=.33\textwidth]{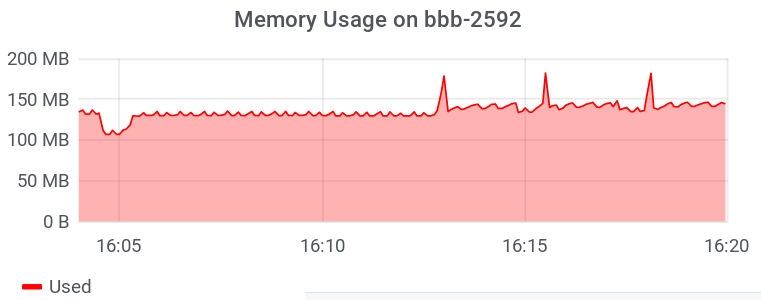}}
\caption{\platform Total Resource usage. Each plot shows the resource consumption on a node.}\label{fig:overhead}
\vspace{-0.1in}
\end{figure*}


In this section, we describe an implementation of \platform. 
The code is available on GitHub~\cite{modicum}. We use a private Ethereum network to provide smart contract functionality. \platform actors, including Job Creators, Resource Providers, Solvers, and Directory services are implemented as Python services. The matching solver uses a greedy approach to match offers as they become available using a maximum bipartite matching algorithm.  These actors use the JSON-RPC interface to connect to the Ethereum Geth client \cite{geth}.  Note that each actor can be configured with any number of Geth clients, and the ledger can be implemented either as a private blockchain or we can use the main Ethereum chain as the ledger. We use Docker \cite{merkel2014docker} images to package jobs. Jobs can be run securely by separating the job from the underlying infrastructure through isolation. This can be achieved using a hardening solution such as AppArmor \cite{Bauer2006AppArmor} or seccomp \cite{seccomp} in conjunction with Docker. This protects computational nodes from erroneous or malicious jobs, if properly configured. Proper configuration has been discussed in other papers, for example \cite{bui2015analysis}, and we do not go into detail here. 
To determine job requirements, jobs are profiled using cAdvisor~\cite{cadvisor}. 



As part of the offer and matching specification, we require the JC to include the hash of the base of the Docker image. Additionally, during setup, Mediators and RPs specify a set of supported base Docker images (\cref{fig:seqdiag}: \textit{mediatorAddFirstLayer} and \textit{resourceProviderAddFirstLayer}). This permits some optimization since RPs are able to specify which base images they have installed, thus by matching them accordingly, we can reduce the bandwidth required to transfer the job by the size of the base image. Common base images vary between about 2MB - 200MB \cite{crunch}. 

\subsection{Experimental Evaluation}
\label{sec:experimental}


 JCs, RPs, and a Mediator were deployed on a 32 node BeagleBone cluster with Ubuntu 18.04. 
 We set up a private Ethereum network.
 The Solver and Directory were  deployed on an  Intel i7 laptop with 24GB RAM.   The actors connect to the Geth client \cite{geth} each using a unique Ethereum account. 

\textbf{Measuring Gas Costs and Function Times:} To measure the minimum cost of executing a job via \platform, we had a single JC submit 100 jobs 
 and measured the function gas costs and call times independent of the job that was being executed.  
These can be found in 
\FullText{\cref{app:gascost}.}{our full paper \cite{modicum-arxiv}.}
The JC's average gas cost of a nominal execution is 
$592,000$ gas. At current Ethereum prices, this converts to \$0.168 per job for the JC~\cite{ethpric}. Comparing this to Amazon Lambda pricing~\cite{awsprice}  on a machine with 512MB RAM (which a BeagleBone has), a job would have to last {\textasciitilde}6 hours to incur the same cost. However, the electrical costs to run a BeagleBone (210-460mA @5V) at maximum load  for that long, assuming \$0.12/kWh electricity price, is only \$0.0016. This illustrates that there is potential for such a transaction system to be a viable option compared to AWS. However, using Ethereum as the underlying mechanism is currently only viable for long running jobs; but work is underway to improve the efficiency of Ethereum~\cite{ethShard}. 
We also measured the gas cost of a mediated execution: the JC's average gas cost in this case is $991,000$, and the Mediator's cost to post the mediation result is $187,000$.

\catcode`\_=12 
\pgfplotstableread[col sep=comma]{data/sim24/acceptResult.csv}\dtAcceptRes
\pgfplotstableread[col sep=comma]{data/sim23/rejectResult.csv}\dtRejectRes
\pgfplotstableread[col sep=comma]{data/sim24/postJobOfferPartOne.csv}\dtPostJOffPA
\pgfplotstableread[col sep=comma]{data/sim24/postJobOfferPartTwo.csv}\dtPostJOffPB
\pgfplotstableread[col sep=comma]{data/sim24/postMatch.csv}\dtPostMatch
\pgfplotstableread[col sep=comma]{data/sim24/postResOffer.csv}\dtPostResOff
\pgfplotstableread[col sep=comma]{data/sim24/postResult.csv}\dtPostResult

\pgfplotstableread[col sep=comma]{data/sim24/SPB.csv}\SPB
\pgfplotstableread[col sep=comma]{data/sim24/SPMC.csv}\SPMC

\catcode`\_=8 

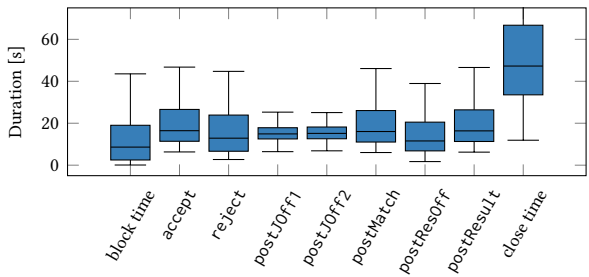
\begin{figure}[t]
\begin{tikzpicture}
  \begin{axis}[
    boxplot/draw direction = y,
    xtick = {1,2,3,4,5,6,7,8,9},
    yticklabel style = {font=\footnotesize},
    xticklabel style = {font=\footnotesize, rotate=60},
    xticklabels = {block time, \texttt{accept}, \texttt{reject}, \texttt{postJOff1}, \texttt{postJOff2}, \texttt{postMatch}, \texttt{postResOff}, \texttt{postResult}, close time},
    ylabel style = {font=\footnotesize},
    ylabel = {Duration [s]},
    ymax = 75, 
    ymin = -5,
    cycle list = {set1-blue},
    height=.45\columnwidth,
    width=\columnwidth
    ]
    \addplot+[boxplot, fill, draw=black] table[ y = {dt}] {\SPB};
    
    \addplot+[boxplot, fill, draw=black] table[ y = {dt}] {\dtAcceptRes};
    \addplot+[boxplot, fill, draw=black] table[ y = {dt}] {\dtRejectRes};
    \addplot+[boxplot, fill, draw=black] table[ y = {dt}] {\dtPostJOffPA};
    \addplot+[boxplot, fill, draw=black] table[ y = {dt}] {\dtPostJOffPB};
    \addplot+[boxplot, fill, draw=black] table[ y = {dt}] {\dtPostMatch};
    \addplot+[boxplot, fill, draw=black] table[ y = {dt}] {\dtPostResOff};
    \addplot+[boxplot, fill, draw=black] table[ y = {dt}] {\dtPostResult};
    
    \addplot+[boxplot, fill, draw=black] table[ y = {dC}] {\SPMC};
  \end{axis}
\end{tikzpicture}
\vspace{-0.8em}
\caption{Duration of \platform function calls during nominal operation.}
\label{fig:dt-run}
\vspace{-0.1in}
\end{figure}


During these tests,
 duration of the function calls was also measured (see \cref{fig:dt-run}). We note that the mean time for a block with transactions to be mined (block time in \cref{fig:dt-run}) is about every 10 seconds, and that function call delay is consistently about 5 to 10 seconds longer. This may be attributable to calls missing a recent block. The close time is the time measured between \texttt{MatchClosed} events. Since the jobs were run sequentially it is a measure of the cumulative time added to the execution of a job, in this test running a job through \platform added approximately 52 seconds. 

{\bf Measuring the Overhead of Platform:} To measure the overhead, 
we compared the execution of jobs run with Docker containers natively against jobs run in \platform. The job we used was the bodytrack computer vision application drawn from the PARSEC benchmarking suite \cite{Bienia2008Parsec}. PARSEC has been used to benchmark resource allocation platforms \cite{Shekhar2017INDICES} as well as platforms for high performance computing \cite{Ren2017Nosv} among others. We again ran 100 jobs, which took a total of 221 minutes, averaging 2 minutes per job. The mean time for a block to be mined was 31 seconds, meaning it took about 4 blocks to complete a job. The block time was likely longer in this experiment because as the blockchain grew block mining times appeared to increase, though we did not study this explicitly. 
This application tracks the 3D pose of a human body through a sequence of images. In Figure \ref{fig:CPUjobCompare}, we see that \platform increases the runtime by about 1 second or 4\%. In Figure \ref{fig:MEMjobCompare}, the average memory consumption increases by 0.1MB, or ~3\%. To check mediation, we ran the jobs again, but rejected the results for all 100 jobs and requested mediation.
\catcode`\_=12 
\pgfplotstableread[col sep=comma]{data/native/runtime.csv}\native
\pgfplotstableread[col sep=comma]{data/sim26/runtime.csv}\modicum

\pgfplotstableread[col sep=comma]{data/native/total_rss.csv}\native
\pgfplotstableread[col sep=comma]{data/sim26/total_rss.csv}\modicum
\catcode`\_=8 

\begin{figure}[t]
\centering
\subfloat[Runtime]{
\begin{tikzpicture}
  \begin{axis}[
    boxplot/draw direction = y,
    xtick = {1,2},
    yticklabel style = {font=\footnotesize},
    xticklabel style = {font=\footnotesize},
    xticklabels = {Native, \platform},
    ylabel style = {font=\footnotesize},
    ylabel = {Runtime [s]},
    ymin=22,
    height=.4\columnwidth,
    width=.5\columnwidth
    ]
    \addplot+[boxplot, fill, draw=black] table[ y = {runtime}] {\native};
    
    \addplot+[boxplot, fill, draw=black] table[ y = {runtime}] {\modicum};
    
  \end{axis}
\end{tikzpicture}
\label{fig:CPUjobCompare}
}
\subfloat[Average memory]{
\begin{tikzpicture}
  \begin{axis}[
    boxplot/draw direction = y,
    xtick = {1,2},
    yticklabel style = {font=\footnotesize},
    xticklabel style = {font=\footnotesize},
    xticklabels = {Native, \platform},
    ylabel style = {font=\footnotesize},
    ylabel = {Memory [MB]},
    height=.4\columnwidth,
    width=.5\columnwidth
    ]
    \addplot+[boxplot, fill, draw=black] table[ y = {mean}] {\native};
    
    \addplot+[boxplot, fill, draw=black] table[ y = {mean}] {\modicum};
    
  \end{axis}
\end{tikzpicture}
\label{fig:MEMjobCompare}
}
\caption{Running time and memory usage on \platform and native execution.}
\label{fig:Compare}
\vspace{-.5cm}
\end{figure}
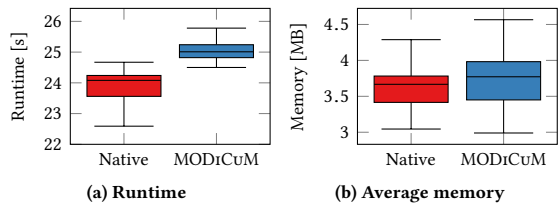

Resource consumption while running the benchmark on \platform can be seen in \cref{fig:overhead}. The nodes are at idle during the valley from 16:04:30 to 16:05:30 at which point the platform is started. From 16:05:30 onward, \platform is running, and the bursts that can be seen, for example at 14:09:30-14:10:45 in Figure \ref{fig:RPCPU_OH}, are when jobs are being executed. \platform introduces about 20-25MB RAM overhead for each agent type. It introduces ~80\% CPU overhead on the Mediator, 30\% on the Job Creator, and no apparent change to the Resource Provider. This is acceptable since the BBB devices are resource constrained and so the overhead will be less significant on more powerful compute nodes.




\section{Conclusions}
\label{sec:concl}
\iflong
\begin{table}
    \centering
    \caption{Summary of the key concepts in \platform}
    \label{tab:summary}
    \resizebox{\columnwidth}{!}{%
        \begin{tabulary}{\textwidth}{|J|J|J|}
            \hline Requirement & Component \\
            \hline transfer funds between agents & Smart contract \\
            \hline fully specify and create offers  & Smart contract \\
            \hline share and find offers  & Smart contract \\
            \hline pair resource and job offers in matches  & Solver \\
            \hline notify the JC and the RP of the match & Smart Contract \\
            \hline transfer data between agents & Directory \\
            \hline raise disputes & Job Creator via smart contract \\
            \hline resolve disputes & Mediator via smart contract \\
            \hline achieve consensus on the state & Ethereum platform \\
            \hline
            
        \end{tabulary}
        }
\end{table}
We described the mechanisms and interactions required to implement an open and decentralized computation market. 
\cref{tab:summary} summarizes the key concepts.
An open market of computational resources, where resource owners and resource users trade directly with each other, can lead to greater participation and more competitive pricing. Our design deters participants from misbehaving---which is a key problem in decentralized systems---by resolving disputes via dedicated mediators and by imposing enforceable fines. 
We provided analytic results showing that it is possible to set up a market where it is in the best interests of rational participants to be honest.
Finally, we presented an implementation of \platform based on Ethereum, Python, and Docker, and demonstrated via experiments the overhead. Detailed specification and implementation of  \platform is available online~\cite{modicum}.
\se{if room include concepts about relaxing limitations.}
\fi

An open market of computational resources, where resource owners and resource users trade directly with each other has the potential for greater participation than volunteer computing and more competitive pricing than cloud computing. The key challenges associated with implementing such a market stem from the fact that any agent can participate and behave maliciously. Thus, mechanisms for detecting misbehavior and for efficiently resolving disputes are required. In this paper we propose a smart contract-based solution to enable such a market. Our design deters participants from misbehaving by resolving disputes via dedicated Mediators and by imposing enforceable fines through the smart contract. This is possible because we recognized that the results do not need to be globally accepted, convincing the JC will often suffice. We learned that due to the limitations of Ethereum our platform is only suitable for long running tasks, but there is space between the cost of electricity and AWS for a platform of this nature. Future work is looking into other platforms that support smart contracts, as well as leveraging improvements to Ethereum. 

\textbf{Acknowledgments}
We are very thankful to Prof. Gabor Karsai and Prof. Aniruddha Gokhale for their insight and comments on the paper. This work was supported in part by National Science Foundation through award numbers 1647015 and 1818901.

\bibliographystyle{ACM-Reference-Format}
\bibliography{references,refs} 
\balance

\onecolumn
\newpage
\twocolumn
\ifFull
    \appendixpage 
    
    \begin{appendices}
    \setcounter{theorem}{0}
\section{Definitions} \label{app:specification}
In this appendix we define what it means for a match to be feasible and the how the deposits and payments are determined.

\subsection{Matching Feasibility}\label{app:feasible}

We can match two offers if they satisfy the following conditions: Job offer limit variables must be lower than the resource offer cap variables. Job offer max variables must be higher than the resource offer price variables.
Additionally, there should be a mediator with the same architecture of JO which is trusted by both JC and RP.


\begin{align}
RO.\textnormal{\textit{instructionCapacity}} &\geq JO.\textnormal{\textit{instructionsLimit}}  \\
RO.\textnormal{\textit{ramCapacity}} &\geq JO.\textnormal{\textit{ramLimit}} \\
RO.\textnormal{\textit{localStorageCapacity}} &\geq JO.\textnormal{\textit{localStorageLimit}}  \\
RO.\textnormal{\textit{bandwidthCapacity}} &\geq JO.\textnormal{\textit{bandwidthLimit}} \\
RO.\textnormal{\textit{instructionPrice}} &\leq JO.\textnormal{\textit{instructionMaxPrice}}  \\
RO.\textnormal{\textit{bandwidthPrice}} &\leq JO.\textnormal{\textit{maxBandwidthPrice}}  \\
JO.\textnormal{\textit{architecture}} &= RP.\textnormal{\textit{architecture}} \\
JO.\textnormal{\textit{directory}} &\in RP.\textnormal{\textit{trustedDirectories}}
\end{align}
\begin{align}
\exists\, M: ~ &M \in (JC.\textnormal{\textit{trustedMediators}} \cap RP.\textnormal{\textit{trustedMediators}}) \nonumber \\ 
& \wedge M.\textnormal{\textit{architecture}} = RP.\textnormal{\textit{architecture}} \\
& \wedge JO.\textnormal{\textit{directory}} \in M.\textnormal{\textit{trustedDirectories}} 
\end{align}

\begin{gather}
    \begin{aligned}
        & \textnormal{\textit{currentTime}}  + RP.\textnormal{\textit{timePerInstruction}} \cdot JO.\textnormal{\textit{instructionLimit}}  \\
        & \leq JO.\textnormal{\textit{completionDeadline}}
    \end{aligned}
\end{gather}


\subsection{Payment}


When the JC or the RP wants to post an offer for a job or resource, it will pay a deposit value to prevent it from cheating on the platform. This deposit value is a function of the posted offer which is more than the price of the job plus mediation. As the price of job or mediation is unclear at the time of posting the offer, a static \texttt{$penaltyRate \gg 1$} is applied as security deposit.

After the job is finished and both the JC and RP agree on the outcome, JC will receive the deposit minus the cost of the job and the RP will receive the deposited value plus the cost of the job.

In case of disagreement, the match will go for mediation. The mediators have a fixed price for their resources. After the submission of their verdict, the deposit of the party who is at fault will be forfeited in favor of the winner and winner will receive both the deposits minus the cost of mediation.

\begin{equation}
    \begin{aligned}
    JO.deposit = & (instructionLimit   \cdot instructionMaxPrice + \\ 
                 & bandwidthLimit \cdot bandwidthMaxPrice) \cdot  \\ 
                 & (\theta + n) + \pi_a 
    \end{aligned}
\end{equation}

\begin{equation}
    \begin{aligned}
    RP.deposit = & (instructionCap \cdot instructionPrice + \\
                 & bandwidthCap \cdot bandwidthPrice ) \cdot \\
                 & (\theta + n)  + \pi_a 
    \end{aligned}
\end{equation}{}

\begin{gather}
    \begin{aligned}
        \pi_c    =& result.instructionCount \cdot \textnormal{\textit{resourceOffer}}.instructionPrice +  \\
                  & result.bandwidthUsage \cdot \textnormal{\textit{resourceOffer}}.bandwidthPrice
    \end{aligned}
    \raisetag{15pt}
\end{gather}

\section{Data Structures}
This appendix includes the definitions of the data structures used to model the components of the platform. 
The platform is a composition of a smart contract, Resource Providers, Job Creators, Mediators, and Directories, and has the following structure:



\label{model:Platform}
\begin{alltt}
Platform\{ mediators: Mediator[],
          resourceProviders: ResourceProvider[],
          jobCreators: JobCreator[],
          resourceOffers: ResourceOffer[],
          jobOffers: JobOffer[],
          matches: Match[],
          results: JobResults[],
          mediatorResults: MediatorResult[],
          penaltyRate : uint \}
\end{alltt}

The platform can support multiple architectures. However, initially it will only have amd64 and armv7.

\label{model:Architecture}
\begin{alltt}
enum Architecture\{ amd64,
                   armv7 \}
\end{alltt}

\subsection{Entities}
The data structure representation of the elements that comprise the  platform are defined below. 

\label{model:JC}
\begin{alltt}
JobCreator\{ trustedMediators: Mediator[]\}
\end{alltt}

\label{model:RP}
\begin{alltt}
ResourceProvider\{ trustedMediators: Mediator[],
                  trustedDirectories: address[],
                  arch: Architecture,
                  timePerInstruction: uint\}
\end{alltt}




\label{model:Mediator}
\begin{alltt}
Mediator\{ arch: Architecture,
          instructionPrice: uint,
          bandwidthPrice: uint,
          trustedDirectories: address[],
          supportedFirstLayers: uint,
          availabilityValue: uint  \}
\end{alltt}



\subsection{Offers}

The format of the job offer is shown below. \texttt{jobCreator} is a unique identifier for the JC. The \texttt{depositValue} is the JC's security deposit for the job. The limits specify how many instructions the Job Creator is willing to pay for (after executing this many, the Resource Provider can give up and still get paid), how much RAM, local storage, and bandwidth the Resource Provider may have to use (again, after reaching these limits, the Resource Provider may stop). \texttt{instructionMaxPrice} specifies the maximum price per instruction that the Job Creator accepts, \texttt{bandwidthMaxPrice} specifies the maximum price per downloaded/uploaded byte (for the job execution layer, not for the base layer) that the Job Creator accepts. \texttt{completionDeadline} is the deadline of RP for submitting the solution. \texttt{matchIncentive} is the reward offered by the JC to the Solver. \texttt{firstLayerHash} is the hash of the job's Docker image. \texttt{URI} gives the location of the job-specific files on the Directory. \texttt{directory} is the identifier of the trusted directory where the job is located. The \texttt{jobHash} is a hash of the Docker image. \texttt{arch} specifies which architectures support the job.

\label{model:JO}
\begin{alltt}
JobOffer\{ depositValue: uint,
          instructionLimit: uint,
          bandwidthLimit: uint,
          instructionMaxPrice: uint,
          bandwidthMaxPrice: uint,
          completionDeadline: uint,
          matchIncentive: uint,
          jobCreator: JobCreator,
          firstLayerHash: uint,
          ramLimit: uint,
          localStorageLimit: uint,
          uri: bytes,
          directory: address[],
          jobHash: uint,
          arch: Architecture;\}
\end{alltt}

\label{model:RO}
\begin{alltt}
ResourceOffer\{ resProvider: address,
               depositValue: uint,
               instructionPrice: uint,
               instructionCap: uint,
               memoryCap: uint,
               localStorageCap: uint,
               bandwidthCap: uint,
               bandwidthPrice: uint,
               matchIncentive: uint,
               verificationCount: uint;\}
\end{alltt}
Prices are per instruction or per byte. Capacities specify what resources the Resource Provider has, and they are used as constraints for matching. All of the participants can deposit as much as they want as long as it is more than the required \texttt{JO/RP.deposit}. \texttt{depositValue} is the variable that holds the amount of deposited value for each offer.
    
\subsection{Results}\label{app:result}
There are various results that are produced during operation of the platform. This appendix describes the data structure of each output. 

When offers are posted, they are recoreded to the ledger. Solvers read the ledger to discover offers and then solve a resource allocation problem to pair compatible offers. The format of the match submitted back to the ledger is below, where \texttt{resourceOffer/jobOffer} is the matched resource/job offer (as defined in \cref{model:RO}), and \texttt{mediator} is a mediator that exists in the trusted list of both the RP and JC data structures (as defined in \cref{model:JC}).
\label{model:Match}
\begin{alltt}
    Match \{ resourceOffer: ResourceOffer,
            jobOffer: JobOffer,
            mediator: Mediator\}
\end{alltt}

The result provided by an RP is structured as a \texttt{JobResult} (see below). \texttt{uri} gives the location of the result on the Directory. \texttt{matchId} is an internal identifier for the match. \texttt{hash} is a hash of the result submitted to the directory. \texttt{instructionCount} is the number of instructions that were executed by the ResourceProvider. \texttt{bandwidthUsage} is the number of bytes downloaded / uploaded by the Resource Provider for the job (not counting the download of Docker layers). The JC has a specific deadline for responding to a result. Whether to approve or decline it. Whether a particular result has been responded to is recored in the \texttt{reaction} filed. The deadline is computed from the \texttt{timestamp} recored in the result. If the deadline is missed the RP can accept the result instead of the JC.

\label{model:JR}
\begin{alltt}
    JobResult\{ uri: bytes,
               matchId: uint,
               hash: uint,
               instructionCount: uint,
               bandwidthUsage: uint,
               reacted: Reaction,
               timestamp: uint,
               status: ResultStatus; \}
\end{alltt}{}

\texttt{status} specifies the specific outcome of the execution and is one of the outcomes defined in the \texttt{ResultStatus} data structure (see below).

\label{model:result}
\begin{alltt}
    enum ResultStatus\{ Completed,
                       Declined,
                       JobDescriptionError,
                       JobNotFound,
                       MemoryExceeded,
                       StorageExceeded,
                       InstructionsExceeded,
                       BandwidthExceeded,
                       ExceptionOccured,
                       DirectoryUnavailable \}
\end{alltt}

\texttt{Completed} means that the Resource Provider finished the job successfully and posted the results on the Directory.
A \texttt{Declined} results status indicates that the RP chose not to execute the job.
\texttt{JobDescriptionError} means that there is an error in the job description.
\texttt{MemoryExceeded}, \texttt{InstructionsExceeded}, \texttt{Bandwidth\-Exceeded} mean that the job exceeded the limits specified in the JobOffer.
\texttt{ExceptionOccured} means that an exception was encountered while executing the job, while \texttt{DirectoryUnavailable} means that the Resource Provider is unable to post results because the Directory is unavailable.

\label{model:MR}
\begin{alltt}
    MediatorResult\{ status: ResultStatus,
                    uri: bytes,
                    matchId: uint,
                    hash: uint,
                    instructionCount: uint,
                    bandwidthUsage: uint,
                    verdict: Verdict,
                    faultyParty: Party; \}
\end{alltt}{}

The platform on itself cannot determine which party is cheating merely based on the \texttt{hash} of the Mediator's Results. For example, the smart contract cannot check whether the job was correctly posted to the directory by the Job Creator or the result was correctly posted to directory by the Resource Provider. Therefore, it is the responsibility of the Mediator to decide who should be punished in the ecosystem.
\texttt{verdict} is the reason for deciding on the cheating party, and \texttt{faultyParty} is the cheating party.

\label{model:verdict}
\begin{alltt}
    enum Verdict\{ ResultNotFound,
                  TooMuchCost,
                  WrongResults,
                  CorrectResults,
                  InvalidResultStatus\}
\end{alltt}
This is the reason that the Mediator chose to punish an actor.
\texttt{ResultNotOnDirectory} means that the Resource Provider did not put the results in the directory.
\texttt{TooMuchCost} means that the Resource Provider charged too much for the job.
\texttt{WrongResults} means that the Resource Provider provided wrong results.
\texttt{CorrectResults} means that the Resource Provider completed the job correctly and the JobCreator should be punished by sending the result for mediation.
\texttt{InvalidResultStatus} means that the Resource Provider's mentioned ResultStatus is wrong. For example, it said that the job description was invalid and returned \texttt{JobDescriptionError}, but the description was correct.

\label{model:party}
\begin{alltt}
    enum Party\{ ResourceProvider, JobCreator \}
\end{alltt}

These are the trustless parties in the ecosystem. Mediators should specify who cheated in a job and should be punished by specifying one of these parties.


\section{Proofs: Equilibrium Analysis}\label{app:proofs}

This appendix provides the proofs for the theorems presented in \cref{sec:Eqanalysis}.
The RP chooses between executing a job and attempting to deceive the JC. The JC chooses between verifying the result and accepting it without verification. Below we examine the payoffs for each combination to determine which conditions will cause each agent to always choose one action (the pure strategy) or to mix randomly between them. The utilities for the RP and JC can be found in \cref{tab:RPU,tab:JCU} respectively. To simplify these utility functions, we set $\pi_{d}$ (the compensation when cheating is detected) equal to $\pi_{c}$ (the amount the JC would have paid the RP for successfully completing a job). This means that if cheating is detected in the RP, the JC will receive at least what it was willing to pay, and if cheating was detected in the JC the RP will receive what it expected for correct execution. This was chosen to limit the additional benefit to dishonest participants while minimizing the harm to honest participants. These simplified utilities can be found in \cref{tab:RPU_pic,tab:JCU_pic}. The proofs for the following theorems rely on determining which strategy dominates thus we further simplify the payoffs by removing terms that are common to both strategies, leaving only the terms that determine which strategy is dominant. The result is \cref{tab:RPDS,tab:JCDS}. 

To refer to the utility of a given combination of actions we use the symbol $U$ with superscript of the agent (i.e., RP and JC) and subscript of the action combination of RP and JC (i.e., EV is <execute,verify>, DP is <deceive,pass>, EP is <execute,pass>, and DV is <deceive,verify>). So for example the utility of the RP when the the RP executes and the JC passes is denoted $U_{EP}^{RP}$ and refers to the upper right utility in \cref{tab:RPU,tab:RPU_pic}. Since \cref{tab:RPDS,tab:JCDS} do not represent utility but rather include only the significant terms for determining strategy dominance we refer to them with $\Delta U$ instead.

\begin{table*}[h]
\centering
\begin{minipage}{.45\textwidth}
    \centering
    \caption{RP payoffs by decision}
    \label{tab:RPU}
    \resizebox{0.9\columnwidth}{!}{%
    \begin{tabular}{@{}lp{5cm}l@{}}
    \toprule
            & verify                                                                            & pass     \\ \midrule
    execute & $\overbrace{- c_{e} - g_{r} - \pi_{a} +}^{U_{EV}^{RP}} \newline 
    \pi_{c} \Big(n p_{a}^{n} \left(p_{a} - 1\right) + p_{a} + p_{a}^{n} \theta \left(p_{a} - 1\right)\Big) +  \newline
    \pi_{d} \left(1-p_{a}\right)\left(1-p_{a}^{n}\right)$ & $\overbrace{\pi_{c}- c_{e} - g_{r} - \pi_{a}}^{U_{EP}^{RP}}$ \\
    \rule{0pt}{15pt} 
    deceive & $- c_{d} - g_{r} - \pi_{a} + \newline \underbrace{p_{a}^{n} \pi_{c} \left(- n - \theta\right) + \pi_{d} \left(1 - p_{a}^{n}\right)}_{U_{DV}^{RP}}$                     & $\underbrace{\pi_{c}- c_{d}- g_{r} - \pi_{a}}_{U_{DP}^{RP}}$ \\ \bottomrule
    \end{tabular}}
\end{minipage}
\hfil
\begin{minipage}{.5\textwidth}
    \centering
    \caption{JC payoffs by decision}
    \label{tab:JCU}
    \resizebox{\columnwidth}{!}{%
    \begin{tabular}{@{}lp{6cm}l@{}}
    \toprule
             & verify                                                                                                               & pass        \\ \midrule
    execute  & $\overbrace{b - g_{j} - \pi_{c} \left(n + \theta\right) \left(1-p_{a}\right)\left(1-p_{a}^n\right) + }^{U_{EV}^{JC}} \newline \left(1 - p_{a}\right) \left(- g_{m} + p_{a}^{n} \pi_{d}\right) - c_{v} - p_{a} \pi_{c} - \pi_{a}  $ & $\overbrace{b - g_{j} - \pi_{a} - \pi_{c}}^{U_{EP}^{JC}}$   \\
    \rule{0pt}{15pt} 
    deceive  & $- c_{v} - g_{j} - g_{m} + p_{a}^{n} \pi_{d} - \pi_{a} \newline \underbrace{- \pi_{c} \left(n - p_{a}^{n} \left(n + \theta\right) + \theta\right)}_{U_{DV}^{JC}}$                    & $\underbrace{- g_{j} - \pi_{a} - \pi_{c}}_{U_{DP}^{JC}}$     \\  
    \bottomrule
    \end{tabular}}
\end{minipage}{}
\end{table*}

\begin{table*}[h]
\centering
\begin{minipage}{.45\textwidth}
    \centering
    \caption{RP payoffs by decision when $\pi_d = \pi_c$}
    \label{tab:RPU_pic}
    \resizebox{\columnwidth}{!}{%
    \begin{tabular}{@{}lp{5cm}l@{}}
    \toprule
            & verify                                                                            & pass     \\ \midrule
    execute & $\overbrace{- c_{e} - g_{r} - \pi_{a} + \pi_{c}}^{U_{EV}^{RP} }+\newline p_{a}^{n} \pi_{c} (p_{a} - 1) \left(n + \theta + 1\right)$ & $\overbrace{\pi_{c}- c_{e} - g_{r} - \pi_{a}}^{U_{EP}^{RP}}$ \\
    \rule{0pt}{15pt} 
    deceive & $\underbrace{- c_{d} - g_{r} - \pi_{a} + \pi_{c} - p_{a}^{n} \pi_{c} \left(n + \theta + 1\right) }_{U_{DV}^{RP}}$                     & $\underbrace{\pi_{c}- c_{d}- g_{r} - \pi_{a}}_{U_{DP}^{RP}}$ \\ 
    \bottomrule
    \end{tabular}}
\end{minipage}
\hfil
\begin{minipage}{.5\textwidth}
    \centering
    \caption{JC payoffs by decision when $\pi_d = \pi_c$}
    \label{tab:JCU_pic}
    \resizebox{\columnwidth}{!}{%
    \begin{tabular}{@{}lp{6cm}l@{}}
    \toprule
             & verify                                                                                                               & pass        \\ \midrule
    execute  & $\overbrace{b - c_{v} - g_{j} - g_{m} \left(1 - p_{a}\right) - \pi_{a} + \pi_{c} \left(1 - 2 p_{a}\right)}^{U_{EV}^{JC}} \newline - \pi_{c} \left(1 - p_{a}\right) \left(1 - p_{a}^{n}\right) \left(n + \theta + 1\right) $ & $\overbrace{b - g_{j} - \pi_{a} - \pi_{c}}^{U_{EP}^{JC}}$   \\
    \rule{0pt}{15pt} 
    deceive  & $- c_{v} - g_{j} - g_{m} - \pi_{a} + \pi_{c} \newline \underbrace{  + \pi_{c} \left(p_{a}^{n} - 1\right) \left(n + \theta + 1\right) }_{U_{DV}^{JC}}$                    & $\underbrace{- g_{j} - \pi_{a} - \pi_{c}}_{U_{DP}^{JC}}$     \\  
    \bottomrule
    \end{tabular}}
\end{minipage}{}
\end{table*}

\begin{table*}[h]
\centering
\begin{minipage}{.45\textwidth}
    \centering
    \caption{Simplify RP utility to assess dominant strategy with $\pi_d=\pi_c$.}
    \label{tab:RPDS}
    \resizebox{0.7\columnwidth}{!}{%
    \begin{tabular}{@{}lll@{}}
    \toprule
            & verify                                                                            & pass     \\ \midrule
    execute & $\overbrace{- c_{e} + p_{a}^{n+1} \pi_{c} \left(n + \theta + 1\right)}^{\Delta U_{EV}^{RP} }$ & $\overbrace{-c_{e}}^{\Delta U_{EP}^{RP}}$ \\
    \rule{0pt}{15pt} 
    deceive & $\underbrace{- c_{d}}_{\Delta U_{DV}^{RP}}$                                    & $\underbrace{-c_{d}}_{\Delta U_{DP}^{RP}}$ \\ 
    \bottomrule
    \end{tabular}}
\end{minipage}
\hfil
\begin{minipage}{.5\textwidth}
    \centering
    \caption{Simplify JC utility to assess dominant strategy with $\pi_d=\pi_c$.}
    \label{tab:JCDS}
    \resizebox{\columnwidth}{!}{%
    \begin{tabular}{@{}lll@{}}
    \toprule
             & verify                                                                                                               & pass        \\ \midrule
    execute  & $\overbrace{-\pi_{c} \left(1 - p_{a}\right) \left(1 - p_{a}^{n}\right) \left(n + \theta + 1\right) + \left(1 - p_{a}\right) \left(- g_{m} + 2 \pi_{c}\right)^{*1,2}}^{\Delta U_{EV}^{JC}} $ & $ \overbrace{c_v^{*3,4}}^{\Delta U_{EP}^{JC}} $   \\
    \rule{0pt}{15pt} 
    deceive  & $\underbrace{2 \pi_{c} - g_{m} - \pi_{c} \left(1 - p_{a}^{n}\right) \left(n + \theta + 1\right)^{*1,3}}_{\Delta U_{DV}^{JC}} $                    & $\underbrace{c_v^{*2,4}}_{\Delta U_{DP}^{JC}}$     \\  
    \bottomrule
    \end{tabular}}
\end{minipage}{}
\end{table*}

\begin{theorem}[JC should not always pass ($p_v>0$)]
If the JC always \emph{passes} (i.e., $p_v=0$), then the RP's best response is to always \emph{deceive} (i.e., $p_e=0$). 
\label{athm:pv}
\end{theorem}

\begin{proof}
    Assume that the JC always passes, then the RP will either execute or deceive. Since $U_{EP}^{RP}<U_{DP}^{RP}$ is always true then the RP will always deceive. This corresponds to the JC utility $U_{DP}^{JC}$ which is always negative, thus if the JC will always pass then it should not participate. Therefore, a JC that is participating will not always pass. 
\end{proof}{}

\begin{theorem}[$p_e>0$]If $p_v>0$ and $p_{a}^{n+1}>\frac{1}{2}$, then a rational RP must execute the jobs with non-zero probability. 
\label{athm:pe}
\end{theorem}

\begin{proof}
If the JC verifies, the RP dominant strategy is to execute if $U_{DV}^{RP} < U_{EV}^{RP}$. Deriving \cref{tab:RPDS} from \cref{tab:RPU_pic} we know that to execute is dominate if

\begin{equation}
    \begin{aligned}
        - c_{e} + p_{a}^{n+1} \pi_{c} \left(n + \theta + 1\right)>-c_{d} \\
        \text{rearrange terms} \quad p_{a}^{n+1} \pi_{c} \left(n + \theta + 1\right)> c_{e} + -c_{d}
    \end{aligned}{}
    \label{appeq:RPDE}
\end{equation}{}

From \cref{tab:symbols}  we know that $n>0$ so then $2\pi_{c} < \pi_{c}(\theta+n+1)$. We also know $c_e < \pi_{r} \leq \pi_{c}$.  Substituting these into \cref{appeq:RPDE} results in: 

\begin{equation}
    \begin{aligned}
         \Big\{ p_{a}^{n+1} \pi_{c} \left(n + \theta + 1\right) \geq 2\pi_{c}p_{a}^{n+1} \Big\} &> \Big\{\pi_{c} \geq c_{e} \geq c_{e}-c_{d}\Big\} \\
         \text{Simplify: } \quad 2\pi_{c}p_{a}^{n+1} &> \pi_{c} \\
         \text{Solve for } p_{a}^{n+1} \text{:} \quad p_{a}^{n+1}&>\frac{1}{2}
    \end{aligned}
    \label{appeq:RPDE2}
\end{equation}{}

Thus, since $p_v>0$ from \cref{athm:pv}, if \cref{appeq:RPDE2} is true then executing is the dominant strategy and $p_{e}>0$. 
\end{proof}{}

\begin{theorem}[Bounded $p_a$]
    If the parameters in the JC utility function are all set to minimize the optimal value for $p_{a}$ (maximizing a rational JC's dishonesty), then any deviation will increase $p_a$. The platform controls $n$ and $\theta$, and so controls the minimum optimal value of $p_a$. 
    \label{athm:bpa}
\end{theorem}{}

\begin{proof}
The system must bound $p_a$ for any set of values in the system, thus we use the JC's total expected utility. 
    \begin{equation}
        \begin{alignedat}{2}
        U^{JC} = & p_vp_eU_{EV}^{JC} + p_v(1-p_e) U_{DV}^{JC} + \\
                 & (1-p_v) p_e U_{EP}^{JC} + (1-p_v) (1-p_e) U_{DP}^{JC} \\
        U^{JC} = & b \sigma_{e} - c_{v} \sigma_{v} - g_{j} + g_{m} \sigma_{v} \left(p_{a} \sigma_{e} - 1\right) + \\
                 &\pi_{c} \sigma_{v} \left(n + \theta + 1\right) \left(- p_{a} p_{a}^{n} \sigma_{e} + p_{a} \sigma_{e} + p_{a}^{n} - 1\right) + \\
                 & 2 \pi_{c} \sigma_{v} - \pi_{c} - 2 p_{a} \pi_{c} \sigma_{e} \sigma_{v} - \pi_{a} 
            \end{alignedat}{}
        \label{eq:adJU}
    \end{equation}{}
take the derivative with respect to $p_a$:
    \begin{equation}
        \begin{alignedat}{2}
            \frac{\partial}{\partial p_a}U^{JC} = & g_{m} \sigma_{e} \sigma_{v} - 2 \pi_{c} \sigma_{e} \sigma_{v} + \\
            & \pi_{c} \sigma_{v} \left(n + \theta + 1\right) \left(- n p_{a}^{n} \sigma_{e} + \frac{n p_{a}^{n}}{p_{a}} - p_{a}^{n} \sigma_{e} + \sigma_{e}\right)
        \end{alignedat}
    \end{equation}{}
        
Set equal to 0 and simplify:
    \begin{equation}
        \begin{alignedat}{2}
            \frac{2 \pi_{c} -g_{m}}{\pi_{c}(n + \theta + 1)} =  1 - p_{a}^{n} + \frac{n p_{a}^{n-1}}{p_e} - n p_{a}^{n}
        \end{alignedat}
        \label{aeq:BPA}
    \end{equation}{}

    \begin{figure}[tb]
        \centering
        \resizebox{\columnwidth}{!}
        {
            \input{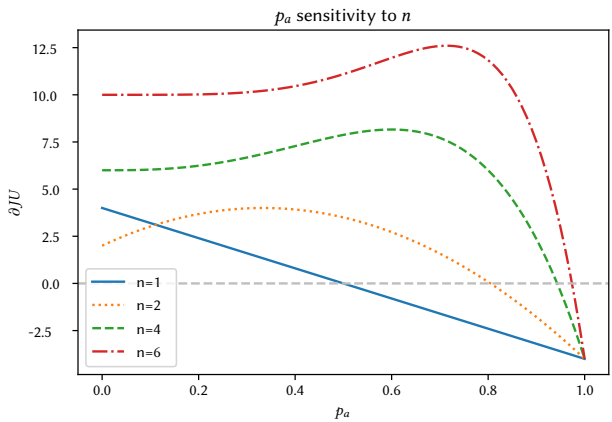}
        }
        \caption{In this plot we vary the value of $n$ and plot $p_a$ against $\frac{\partial U^{JC}}{\partial p_a}$. This shows that as $n$ increases so does the optimal value of $p_a$. For this plot $\pi_{c}=2$, $g_{m}=0$, $\theta=0$, $c_v=1$, $b=4$, $p_e=1$ }
        \label{fig:dJUvsPa}
    \end{figure}

We now examine each term in \cref{aeq:BPA} to determine which values will minimize $p_{a}$.
\begin{itemize}
    \item If $g_{m}$ increases, then left-hand side decreases, meaning $p_a$ will increase.
    \item If $p_e$ increases, then right-hand side decreases, meaning $p_a$ will increase.
    \item $\pi_c$ has no impact on $p_a$ when $g_{m}=0$. If $g_{m}\neq0$, then when $\pi_c$ increases, the left-hand side (lhs) increases and approaches $\frac{2}{\theta+n+1}$ (approaching parity with $g_{m}=0$), then the right-hand side (rhs) must also increase meaning that $p_a$ must decrease. 
    \item If $\theta$ increases, lhs decreases meaning $p_a$ will increase.
    \item If $n$ increases $p_a$ increases, see where the curves cross $0$ in \cref{fig:dJUvsPa}.
\end{itemize}

By assuming a value for each parameter such that any change results in $p_a$ increases we can determine the worst case for the system. Specifically, if the parameter and $p_a$ are inversely related set the parameter to its maximum allowed value, and if they are directly related set the parameter to its minimum value. Thus, the worst case values for each of the parameters are: $g_{m}=0$, $p_{e}=1$, $\theta=0$, $n=1$. The plot in \cref{fig:dJUvsPa1} when $n=1$ uses those parameters and shows that increasing $n$ does cause the optimal value for $p_a$ to increase. We see that when $n=1$, the optimal $p_a=0.5$; and when $n=4$, the optimal $p_a=0.943$. Thus, we see that setting $n$ can set a lower bound on the optimal $p_a$. Parameter $\theta$ also can adjust the lower bound on the optimal value for $p_a$.
    
\end{proof}

\begin{theorem}[JC type 1]  If the JC is type 1, it will always verify ($p_v=1$) since $c_v$ is sufficiently low. This results in a pure strategy equilibrium <execute,verify>.
\label{athm:T1}
\end{theorem}
\begin{proof}
    Type 1 means that $U_{EV}^{JC}>U_{EP}^{JC}$ and $U_{DV}^{JC}>U_{DP}^{JC}$ so verify strictly dominates passing for the JC. From \cref{athm:bpa} we know that $p_{a}^{n+1}>\frac{1}{2}$ and from \cref{athm:pe} we know that if $p_{a}^{n+1}>\frac{1}{2}$ and the JC verifies then the RP will execute thus <execute,verify> is a pure strategy equilibrium. 
\end{proof}

\begin{theorem}[JC type 2]  If the JC is type 2, it results in two pure strategy equilibria <execute,verify>, <deceive,pass>, and one mixed strategy Nash equilibrium where the JC randomly mixes between verifying and passing.
\end{theorem}
\begin{proof}
    Type 2 means that verify dominates when the RP executes ($U_{EV}^{JC}>U_{EP}^{JC}$) and pass dominates when RP deceives ($U_{DV}^{JC}<U_{DP}^{JC}$). Following \cref{athm:T1} we know that <execute,verify> is a pure strategy equilibrium if the JC always verifies and since $U_{EV}^{JC}>U_{EP}^{JC}$ the JC will not choose to pass. 
    
    <deceive,pass> is a pure strategy equilibrium because $U_{DV}^{JC}<U_{DP}^{JC}$ so the JC will not change to verify, and RP will not change to execute because $U_{EP}^{RP}<U_{DP}^{RP}$ is always true.
    
    The JC will only verify all jobs if the cost of verification $c_v$ is sufficiently low ($0$). Otherwise it may get better utility by adopting a mixed strategy, choosing to verify with probability $p_v$ and choosing to pass with probability $1-p_v$. However, if the JC passes, the RP will prefer to deceive (assuming that $p_a^{n+1}>\frac{1}{2}$) and so will also choose to adopt a mixed strategy where it executes with probability $p_{e}$ and deceives with probability $1-p_{e}$.

\end{proof}{}

\begin{theorem}[JC type 3]  If the JC is type 3, it will result in a Nash equilibrium where the JC randomly mixes between verifying and passing, and the RP mixes between executing and deceiving.
\end{theorem}
\begin{proof}
    Type 3 means that pass dominates when the RP executes ($U_{EV}^{JC}<U_{EP}^{JC}$) and verify dominates when RP deceives ($U_{DV}^{JC}>U_{DP}^{JC}$).
    Assume the JC will verify and the RP will execute: <execute,verify>. Then since $U_{EP}^{JC}>U_{EV}^{JC}$ the JC will transition to pass: <execute,pass>. Now since $U_{EP}^{RP}<U_{DP}^{RP}$ is always true the RP will transition to deceive: <deceive,pass>. Then the JC will transition to verify since $U_{DV}^{JC}>U_{DP}^{JC}$: <deceive,verify>. And finally from \cref{athm:pe} the RP will transition to execute:  <execute,verify>. Thus, we return to the initial condition and there is not pure strategy equilibrium. From \cite{NashNon1951} we know that all finite games have an equilibrium and so there must exist a mixed strategy equilibrium. 
\end{proof}{}

\section{Smart Contract Functions}\label{app:SCFunc}

\cref{tab:functions} shows the different smart contract functions and the actor who calls them. the first two functions are called during deployment of the smart contract and these have no subscriber available (N/A).

\section{Function Gas Costs}\label{app:gascost}


\begin{table*}[t]
    \centering
    \caption{Functions provided by the \platform Smart Contract and their corresponding events, as well as the agents who call the functions and the agents that subscribe to those events. In the \emph{Caller} and \emph{Subscriber} columns, mediator is denoted M, Resource Provider is denoted RP, Job Creator is denoted JC, solver is denoted as S.}
    \label{tab:functions}
    
\begin{tabulary}{\textwidth}{|J|J|J|J|}
    \hline
    
    Function & Event & Caller & Subscriber \\
    \hline \texttt{setPenaltyRate}                          & \texttt{PenaltyRateSet}                            & Deployment  & N/A        \\
    \hline \texttt{setReactionDeadline}                     & \texttt{ReactionDeadlineSet}                       & Deployment  & N/A        \\
    \hline \texttt{registerMediator}                        & \texttt{MediatorRegistered}                        & M           & M,S,JC,RP \\
    \hline \texttt{mediatorAddTrustedDirectory}             & \texttt{MediatorAddedTrustedDirectory}             & M           & M,S       \\
    \hline \texttt{mediatorAddSupportedFirstLayer}          & \texttt{MediatorAddedSupportedFirstLayer}          & M           & M,S       \\
    \hline \texttt{registerResourceProvider}                & \texttt{ResourceProviderRegistered}                & RP          & RP,S      \\
    \hline \texttt{resourceProviderAddTrustedMediator}      & \texttt{ResourceProviderAddedTrustedMediator}      & RP          & RP,S      \\
    \hline \texttt{resourceProviderAddTrustedDirectory}     & \texttt{ResourceProviderAddedTrustedDirectory}     & RP          & RP,S      \\
    \hline \texttt{resourceProviderAddSupportedFirstLayer}  & \texttt{ResourceProviderAddedSupportedFirstLayer}  & RP          & RP,S      \\
    \hline \texttt{registerJobCreator}                      & \texttt{JobCreatorRegistered}                      & JC          & JC,S      \\
    \hline \texttt{jobCreatorAddTrustedMediator}            & \texttt{JobCreatorAddedTrustedMediator}            & JC          & JC,S      \\
    \hline \texttt{postResOffer}                            & \texttt{ResourceOfferPosted}                       & RP          & RP,S      \\
    \hline \texttt{postJobOffer}                            & \texttt{JobOfferPosted}                            & JC          & JC,S      \\
    \hline \texttt{cancelJobOffer}                          & \texttt{JobOfferCanceled}                          & JC          & JC,S      \\
    \hline \texttt{cancelResOffer}                          & \texttt{ResourceOfferCanceled}                     & RP          & RP,S      \\
    \hline \texttt{postMatch}                               & \texttt{Matched}                                   & S           & S,RP      \\
    \hline \texttt{postResult}                              & \texttt{ResultPosted}                              & RP          & RP,JC     \\
    \hline \texttt{rejectResult}                            & \texttt{ResultReaction,JobAssignedForMediation}    & JC          & JC,M,RP   \\
    \hline \texttt{acceptResult }                           & \texttt{ResultReaction,MatchClosed}                & JC          & JC,RP     \\
    \hline \texttt{postMediationResult}                     & \texttt{MediationResultPosted,MatchClosed}         & M           & M,JC,RP   \\
    \hline \texttt{timeout}                                 & \texttt{MatchClosed}                               & JC          & JC,RP     \\
    
    \hline
    
\end{tabulary}

\end{table*}


\begin{table*}[!ht]
\caption{Gas cost for each platform function call, measured upon receipt of event emitted by that function. gasEstimate is the value returned when calling the eth\_estimateGas method}
\label{tab:gas_cost}
\begin{tabular}{@{}lrrr@{}}
\toprule
Function                 & mean     & std.dev. & gasEstimate \\ \midrule
setPenaltyRate           &  43,400  &  0      &   22,100     \\        
setReactionDeadline      &  43,100  &  0      &   21,600     \\        
acceptResult             &  140,000 &  64,700  &   62,900     \\        
postResult               &  91,000  &  16,000  &   225,000    \\            
postMatch                &  35,000  &  1,700   &   127,000    \\            
postJobOfferPartOne      &  236,000 &  4,590   &   228,000    \\            
postJobOfferPartTwo      &  216,000 &  0      &   227,000    \\            
postResOffer             &  242,000 &  2,160   &   232,000    \\            
registerMediator         &  137,000 &  0      &   85,100     \\        
registerResourceProvide, &  81,200  &  0      &   25,400     \\        
resourceProviderAddTM    &  65,000  &  0      &   42,300     \\        
registerJobCreator       &  34,000  &  0      &   3,200      \\        
jobCreatorAddTM          &  65,300  &  0      &   42,600     \\        
postMediationResult      &  187,000 &  15,600  &   240,000    \\            
rejectResult             &  53,900  &  9      &   47,100     \\        
nominalcost              &  592,000 &         &             \\  \bottomrule  
\end{tabular}
\end{table*}

To measure the minimum cost of executing a job via \platform, we had a single JC submit 100 jobs and measured the function gas costs and call times independent of the job that was being executed. \cref{tab:gas_cost} shows the cost measured for each function call. The \texttt{postJobOffer} function was split into two parts, \texttt{postJobOfferPart1} and \texttt{postJobOfferPart2} to deal with stack too deep issues. The cost accrued when there are no deviations from the protcol is the cost to post the job and accept the result. This is the nominal cost of operation and is : 
\begin{gather}
    \begin{aligned}
        \text{nominal cost} = & \texttt{postjobOfferPartOne} + \texttt{postjobOfferPartTwo} + \\ 
                              & \texttt{accept\-Result} \\
        \text{nominal cost} = & 2.36E5 + 2.16E5 + 1.40E5 \\
        \text{nominal cost} = & 5.92E5
    \end{aligned}
    \raisetag{20pt}
\end{gather}


    \end{appendices}
\fi

\end{document}